\documentclass[aps, pre,onecolumn]{revtex4-2}
\usepackage{tikz}
\usetikzlibrary{arrows, positioning}
\usepackage{graphicx} 
\usepackage{amsmath,mathtools}
\usepackage{amsthm}
\usepackage{amsfonts}
\usepackage{todonotes}
\usepackage{booktabs} 
\usepackage[margin=0.5in, right=0.5in]{geometry}
\usepackage{hyperref}
\usepackage{comment}
\usepackage{subcaption}
\usepackage[utf8]{inputenc}
\usepackage{tcolorbox}
\usepackage{xcolor}
\usepackage[capitalise]{cleveref}
\usepackage{thmtools}
\usepackage{thm-restate}
\usepackage{appendix}
\usepackage[justification=raggedright,singlelinecheck=true]{caption}
\newtheorem{theorem}{Theorem}

\newtheorem{definition}[theorem]{Definition}

\newtheorem{example}[theorem]{Example}
\newtheorem{conjecture}{Conjecture}
\definecolor{commentcolor}{RGB}{30,144,255} 
\definecolor{replycolor}{RGB}{34,139,34} 

\tcbuselibrary{skins, breakable}

\newtcolorbox{commentbox}{
    breakable,
    colback=white,
    colframe=commentcolor,
    boxrule=0.5mm,
    arc=2mm,
    left=2mm,
    right=2mm,
    before=\vspace{10pt},
    after=\vspace{10pt},
}

\newtcolorbox{replybox}{
    breakable,
    colback=white,
    colframe=replycolor,
    boxrule=0.5mm,
    arc=2mm,
    left=6mm,
    right=2mm,
    before=\vspace{10pt},
    after=\vspace{10pt}
}

\begin{document}
\title{Zero Temperature Dynamics of Ising Systems on Hypercubes}

\author{R.~Chen}
\affiliation{New York University, New York, NY 10012, USA}
\affiliation{Department of Mathematics, University of California, San Diego, La Jolla CA 92093, USA. (Starting on Sep 22nd 2025)}

\author{J.~Machta}
\affiliation{Department of Physics, University of Massachusetts,
            Amherst, MA 01003, USA}
\affiliation{Santa Fe Institute, 1399 Hyde Park Road, Santa Fe, NM 87501, USA}

\author{C.~M.~Newman}
\affiliation{Courant Institute of Mathematical Sciences,
            New York University, New York, NY 10012, USA}
\author{D.~L.~Stein}
\affiliation{Department of Physics and Courant Institute of Mathematical Sciences,
            New York University, New York, NY 10012, USA}

\begin{abstract}
We study the zero-temperature Glauber dynamics of homogeneous Ising ferromagnets on hypercubes, as their dimension $d$ varies. We investigate the asymptotic ($d\to\infty$ and time $t\rightarrow\infty$) behavior of various quantities on hypercubes, such as the final magnetization, the probability for the system to enter a ground state, etc. Our numerical studies are carried out using a uniformly random initial state but with the constraint that
the initial magnetization is zero.  The final states can be divided into three categories: ground states, frozen states, and blinker states. We use the notion of a k-core to describe the geometry of the frozen states and give an exponential lower bound for the number of frozen states in terms of $d$. Blinker states ---  which exist only in even $d$ --- are final states containing at least one blinker (a permanently flipping spin). Blinkers states can have rich local structures; we give explicit constructions for configurations that contain blinkers and prove that the lowest possible dimension for blinker configurations is $d=8$. We also study the `Nature vs. Nurture' problem on hypercubes, asking how much the final state depends on the information contained in the initial configuration, and how much depends on the realization of the dynamical evolution. Finally, we provide several conjectures and suggest some open problems based on the numerical results.
\end{abstract}
\maketitle
\section{Introduction}
\noindent For the past 30 years, the zero-temperature Glauber dynamics of Ising systems have been studied theoretically and numerically on different kinds of graphs, especially $Z^d$. For homogeneous Ising ferromagnets on $Z^d$ (with nearest-neighbor edges), it has been proven~\cite{Newman2000a,Newman2000b} that if the limit $L\rightarrow\infty$ is taken before the limit $t\rightarrow\infty$,  every spin will flip infinitely often for $d\leq 2$; the asymptotic behavior for $d\geq 3$ remains unknown. On the other hand, if the limit $t\rightarrow\infty$ is taken before the limit $L\rightarrow\infty$~\cite{Spirin2001a,Spirin2001b, Olejarz2011}, when $d=1$ the system will always enter a ground state. For $d=2$, the system can either evolve toward a ground state or else (with probability about $0.339$) enter a ``frozen'' state consisting of vertical or horizontal stripes \cite{Kipton2009}.  For $d\geq 3$, the number of frozen states grows extensively, and there exist permanently flipping spins in certain regions \cite{Olejarz2011}.\\\newline
In this paper, we study the zero-temperature dynamics of homogeneous Ising ferromagnets on hypercubes, as the dimension $d$ varies. Because $L$ is fixed on hypercubes, there is no thermodynamic limit. There are two  properties of hypercubes that aid their analysis: the first is self-recursiveness --- two $d$-dimensional hypercubes can be joined to create a $(d+1)$-dimensional hypercube. The second is symmetry --- every site is essentially the same. We find that many dynamical simulations end up in frozen states, which trap the dynamics and prevent the hypercube from entering a ground state. Another interesting feature is the difference in behavior between odd and even dimensions, arising from the presence of zero-energy spin flips in even dimensions and their absence in odd dimensions. Because of this, even-dimensional hypercubes can display permanently flipping spins, which we call `blinkers', which are absent in odd dimensions.\\\newline
There are two types of randomness in our model, that of the initial spin configuration and that of the realization of the dynamics. We are particularly interested in the `Nature vs. Nurture' problem~\cite{Ye2013,Ye2017}, in which one asks how much of the state at time $t$ can be predicted from the information contained in the initial configuration and how much depends on the specific dynamical evolution realization, all as a function of dimension~$d$.\\\newline
This paper is organized as follows. In Section 2, the model and methods for studying it are introduced. In Section 3, we present various numerical results on hypercube dynamics, in particular their long-time behavior as a function of dimension. Section 4 focuses on theoretical results on the geometries of frozen and blinker states; in particular, we give an exponential lower bound for the number of frozen states, provide explicit constructions for several blinker configurations, and show that the lowest dimension for blinker configurations is $8$. In Section 5, we propose several conjectures based on the numerical results, provide some supporting heuristics, and suggest a number of open problems.
\section{Model and Method}
\noindent We consider a homogenous ferromagnetic Ising system $\sigma_i = \pm 1$ on the sites $i$ of the hypercube $Q_d=\{0,1\}^d$, where the system Hamiltonian for a spin configuration $\sigma\in\{-1,1\}^{Q_d}$ is:
\begin{equation}
    \mathcal{H} = -\sum_{\langle i,j\rangle}\sigma_i\sigma_j,
\end{equation}
with the sum over the nearest-neighbor pairs. We denote by $\sigma^d$ a spin configuration on the $d$ dimensional hypercube; in particular, $\sigma^{d,+}\text{ and }\sigma^{d,-}$ are the spin configurations of the ground states (i.e., all spins +1 or all spins -1) on the d-dimensional hypercube.\\\newline
The system is initialized with a random initial configuration constrained to have zero magnetization: each spin is either +1 or -1 with probability $\frac{1}{2}$ but subject to the constraint $\sum_i\sigma_i=0$. The system then evolves according to zero-temperature Glauber dynamics: a site $i$ is chosen uniformly at random, and the energy change $\Delta E_i$ caused by flipping $\sigma_i$  is computed. If $\Delta E_i < 0$, the spin is flipped, and if $\Delta E_i > 0$ the spin is not flipped. If $\Delta E_i = 0$, the spin is flipped with probability $\frac{1}{2}$.
\\\newline
We adopt the algorithm used in \cite{Ye2013} to accelerate the dynamics. Before time $t_0$, the system evolves according to the usual Glauber dynamics described above, with the time incremented by $\frac{1}{2^d}$ after each attempted spin flip. After a time $t_0$, some sites remain active (i.e., their $\Delta E_i \leq 0$), so in order to speed up the simulation, we implement a kinetic Monte-Carlo method: we create a list $a(t)$ of active sites, and then randomly choose a site from the active list. Now time is increased with $\Delta E_i \leq 0$ at each step by $\frac{1}{\left|a(t)\right|}$, where $\left|a(t)\right|$ is the length of the active list. If a spin flips, we need to update the active list by removing (unless $\Delta E_i= 0$) the flipped spin and checking whether its neighbors must be added or removed from the active list.\\\newline
There are three types of final states: two ground states (all spins either $+1$ or all $-1$), frozen states (non-ground-state spin configurations with all spins having $\Delta E_i>0$), and blinker states, in which every spin is either frozen or else flips forever with $\Delta E_i =0$. Blinker spins have previously been found in ferromagnets on Euclidean lattices with $d\geq 3$~\cite{Olejarz2011}. In the hypercube, blinker states can exist only in even dimensions; zero-energy flips ($\Delta E_i = 0$) are not possible in odd dimensions. We denote a spin $\sigma_i$ as a $1$-blinker if all its neighbors are permanently frozen, with exactly half $+1$ and half $-1$. There are many other interesting structures for blinkers, which we will return to in Section 4.
\\\newline
The existence of blinkers complicates the appropriate definition of the limit $t\rightarrow \infty$.  To define final states allowing for blinkers we introduce the idea of all dynamically possible futures.  
Let $G(t)$  be the set of all spin configurations that are dynamically accessible starting from the spin configuration $\sigma(t)$ after running for arbitrarily long times.  By construction, $|G(t)|$ is non-increasing so the limit $G_\infty=\lim_{t\rightarrow\infty}G(t)$ exists.  For ground states and frozen states, $G_\infty$ consists of a single spin configuration, but for blinker states $|G_\infty|>1$.  Since at zero temperature the energy is non-increasing, all configurations in $G_\infty$ must have the same energy, and the dynamical transitions between these configurations must be reversible.  Thus, the dynamics induce an undirected, connected graph, $\mathcal{G}=(G_\infty,V_\infty)$ where the edges $V_\infty$ of the graph correspond to the possible energy-conserving single-spin-flip transitions between the configurations in $G_\infty$.  Every configuration in $G_\infty$ will be visited infinitely often and with equal probability.  Associated with each spin configuration in $G_\infty$ is an active list of spins, all of which can be flipped without energy change.  The union of all of these active lists is the set of blinkers in the final state, which we refer to as $g_\infty$. \\\newline
In Appendix \ref{app:finalstate}, we present a final state algorithm that can be run at any dynamical time $t$, which determines whether or not the system is in a final state or not, i.e., whether $G(t) = G_\infty$. However, for the simulations reported below, we choose a simpler approach. We first run the kinetic Monte Carlo method for a sufficient time~$t$ to ensure that only spins with zero energy change are contained in the active list $a(t)$. If $a(t)$ is empty, the system has entered either a ground state or a frozen state. If not, we find $g_\infty$ using a Depth First Search (DFS) algorithm. Starting with each spin $\sigma_i \in a(t)$, we examine its neighbors $\sigma_j$. For each neighbor, we check its energy change $\Delta E_j$. If $\Delta E_j \leq 2$, $\sigma_j$ is a blinker because it can either flip with zero energy cost immediately or after $\sigma_i$ flips. We then recursively apply this process to the neighbors of $\sigma_j$. The DFS stops when all neighbors of the current spin have $\Delta E \geq 2$, as they cannot flip without increasing energy. This approach ensures that all blinkers are efficiently identified through systematic exploration.\\\newline
Fortunately, as we will see in Section~3, the fraction of blinkers is relatively small. We investigate the following observables as functions of dimension $d$, where the sums exclude the blinkers. We begin with
\begin{equation}
    M_\infty(d) = \lim_{t\rightarrow\infty}\sum_{i\notin g^d_\infty}\sigma^d_i(t),
\end{equation}
\noindent the magnetization of the final configuration excluding the blinkers.
We also study
\begin{equation}
    \mathbb{P}_{\sigma(0),\omega}(\{G^d_\infty =\sigma^{d,+}\text{ or }\sigma^{d,-}\}),
\end{equation}
\noindent the probability for a $d$-dimensional hypercube to enter a ground state as $t\rightarrow\infty$ (as indicated, the probability is defined with respect to the initial configuration~$\sigma(0)$ and the dynamical realization $\omega$).
\\
\noindent Next we study the distribution of the local fields, defined here as the sum of the neighboring spins of $\sigma^d_i$ normalized by $\frac{1}{d}$ but excluding the blinkers (note that by the symmetry of hypercubes, the distribution of $m^i_\infty(d)$ is independent of the choice of $i$, hence we can set $i$ to be the origin $\bar{0}$, and write $m_\infty(d) = m^{\bar{0}}_\infty(d)$ for simplicity):
\begin{equation}
\label{eq:m}
    m^i_\infty(d) = \lim_{t\rightarrow\infty}\frac{1}{d}\sum_{\langle i,j\rangle,j\notin g^d_\infty}\sigma^d_j(t)\, .
\end{equation}
\noindent 
Finally, we study the energy of the final configuration: \begin{equation}
\mathcal{H}_\infty(d) = -\lim_{t\rightarrow\infty}\sum_{\langle i,j\rangle}\sigma^d_i(t)\sigma^d_j(t),    
\end{equation}
\noindent as well as $\overline{f(d)}$, the average number of flips per site excluding the blinkers.\\\newline
\noindent To better understand the nature of the frozen states, we use the concept of the largest $k$-core. We first split the final configuration of a $d$-dimensional hypercube into two graphs, one consisting of only $+1$ spins, and one consisting of only $-1$ spins, preserving the edges that connect spins with the same signs and discarding the edges that connect spins with opposite signs. In each graph, any connected subgraph with all vertices having degree larger than $k\geq\lfloor \frac{d}{2}\rfloor+1$ is frozen. We now ask, what is the size of the largest of these subgraphs? This is the problem of finding the largest $k$-core, where a $k$-core denotes a connected subgraph with at least $k$ neighbors for each vertex in it. We define $\kappa(d) = \max\{k_+,k_-\}$, where $k_+$ (resp., $k_-$) is the core number for the largest $k$-core in the $+1$ (resp., $-1$) graph.
\\\newline
 We  investigate the following two probabilities associated with blinkers.
 \begin{equation}
     \mathbb{P}_{\sigma(0),\omega}(\{|G^d_\infty| > 1\}),
 \end{equation}
\noindent  the probability to enter a blinker state, and 
  \begin{equation}
      \mathbb{P}_{\sigma(0),\omega}(\{\sigma^d_i \in g^d_\infty\}),
  \end{equation}
\noindent   the probability for spin $\sigma^d_i$ to eventually become a blinker (note that again by symmetry, this probability is independent of the choice of $i$).
\\\newline
We also study double-copy dynamics, in which we adopt the method used in \cite{Ye2013}, where the authors introduced the notion of `heritability' to study the so-called nature vs.~nurture problem. Here we initialize a pair of systems
with the same initial configuration and let each system (`twin') evolve independently according to Glauber dynamics. At any time step $t$, we introduce the quantity
\begin{equation}
    q_t(d) = \sum_{i\notin g^d_\infty} \frac{1}{2^d}\sigma^1_i(t)\sigma^2_i(t) ,
\end{equation} 
\noindent which describes the overlap between the two twins, where $S_i^k(t), k=1,2$ stands for the state of the $i$th spin in twin k at timestep t. We study the average $\overline{q_t(d)}$ over both the initial configuration and dynamics, and examine 
\begin{equation}
    \overline{q_\infty(d)} = \lim_{t\rightarrow\infty}\overline{q_t(d)},
\end{equation}
\noindent which is the overlap of the twins in their final states. We are interested in the asymptotic behavior of $\overline{q_\infty(d)}$  when $d\rightarrow\infty$. If it goes to $1$, then the intial configuration completely determines the final state (nature wins); if it goes to $0$, then the final state is completely determined by the dynamical history (nurture wins). Of course, the answer could be strictly between 0 and 1.
\section{Results}
\subsection{Single-copy Dynamics}
\noindent We begin with numerical results for the magnetization of the final state~$M_\infty(d)$. Figure \ref{fig:Figure1} shows the sample distribution of $M_\infty(d)$ scaled by $N$ for different values of $d$. The most salient feature is the difference between the distributions in even and odd dimensions: in even dimensions distributions are generally broader. But as $d\rightarrow \infty$, the distributions for both even and odd dimensions concentrate towards zero, 
hinting at a possible central-limit-type behavior as $d\rightarrow\infty$.
\begin{figure}[h!]
    \centering
    \includegraphics[scale=0.38]{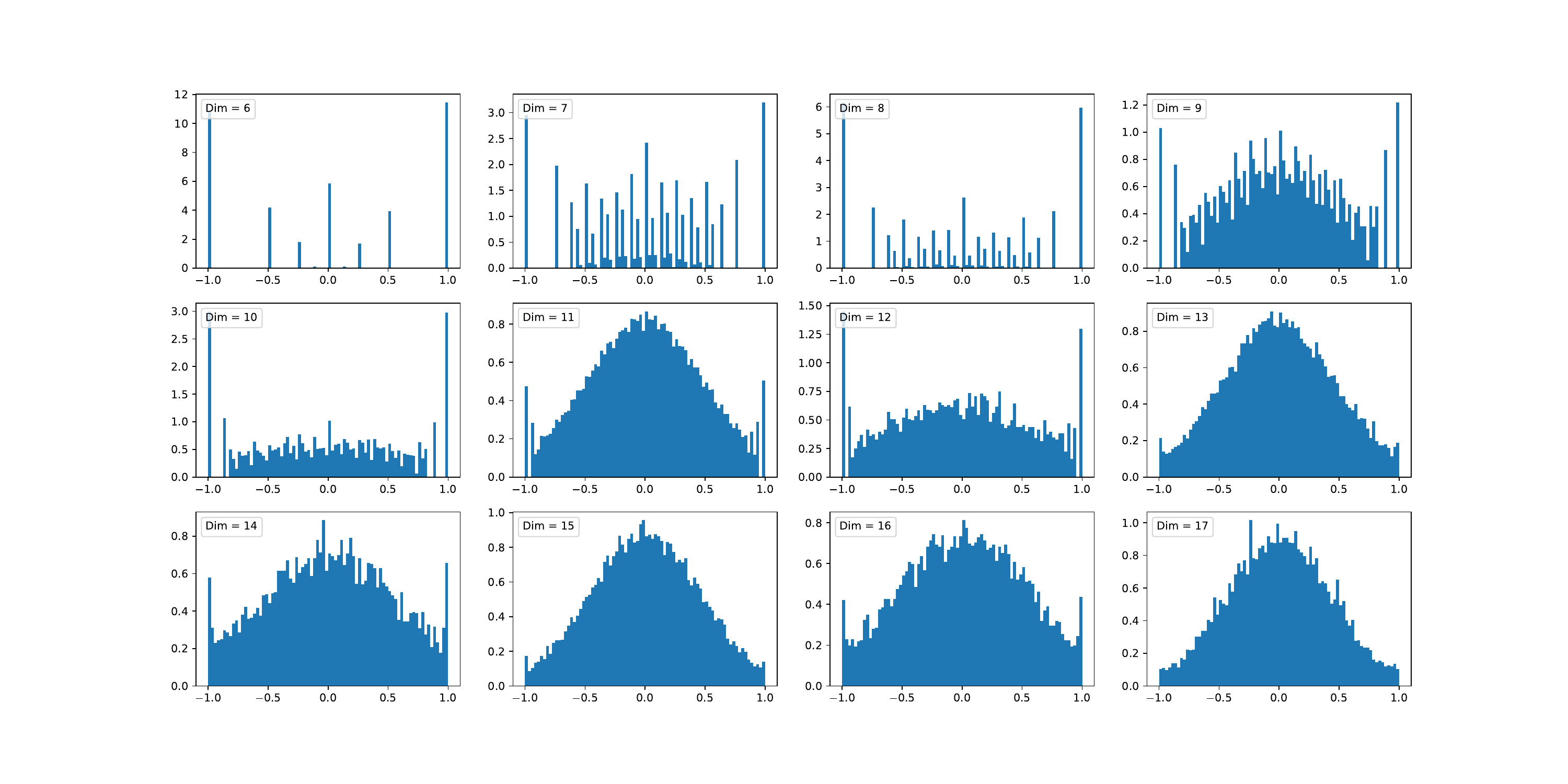}
    \caption{Sample distribution of $M_\infty(d)/N$ at different dimensions}
    \label{fig:Figure1}
\end{figure}
\\\newline
The sum of the two peaks at the ends of the distribution give $\mathbb{P}_{\sigma(0),\omega}(\{G^d_\infty = \sigma^{d,+}\text{ or }\sigma^{d,-}\})$, the probability for the hypercube to enter a ground state as $t\rightarrow\infty$. Figure \ref{fig:Figure2} shows its estimator $\hat{\mathbb{P}}_{\sigma(0),\omega}(\{G^d_\infty = \sigma^{d,+}\text{ or }\sigma^{d,-}\})$ as a function of dimension.
$\hat{\mathbb{P}}_{\sigma(0),
\omega}(\{G^d_\infty = \sigma^{d,+}\text{ or }\sigma^{d,-}\})$ is larger for even dimensions, but decays to zero exponentially fast in dimension for both even and odd dimensions, again consistent with central-limit-type behavior.
\begin{figure}[h!]
    \centering
    \includegraphics[scale=0.40]{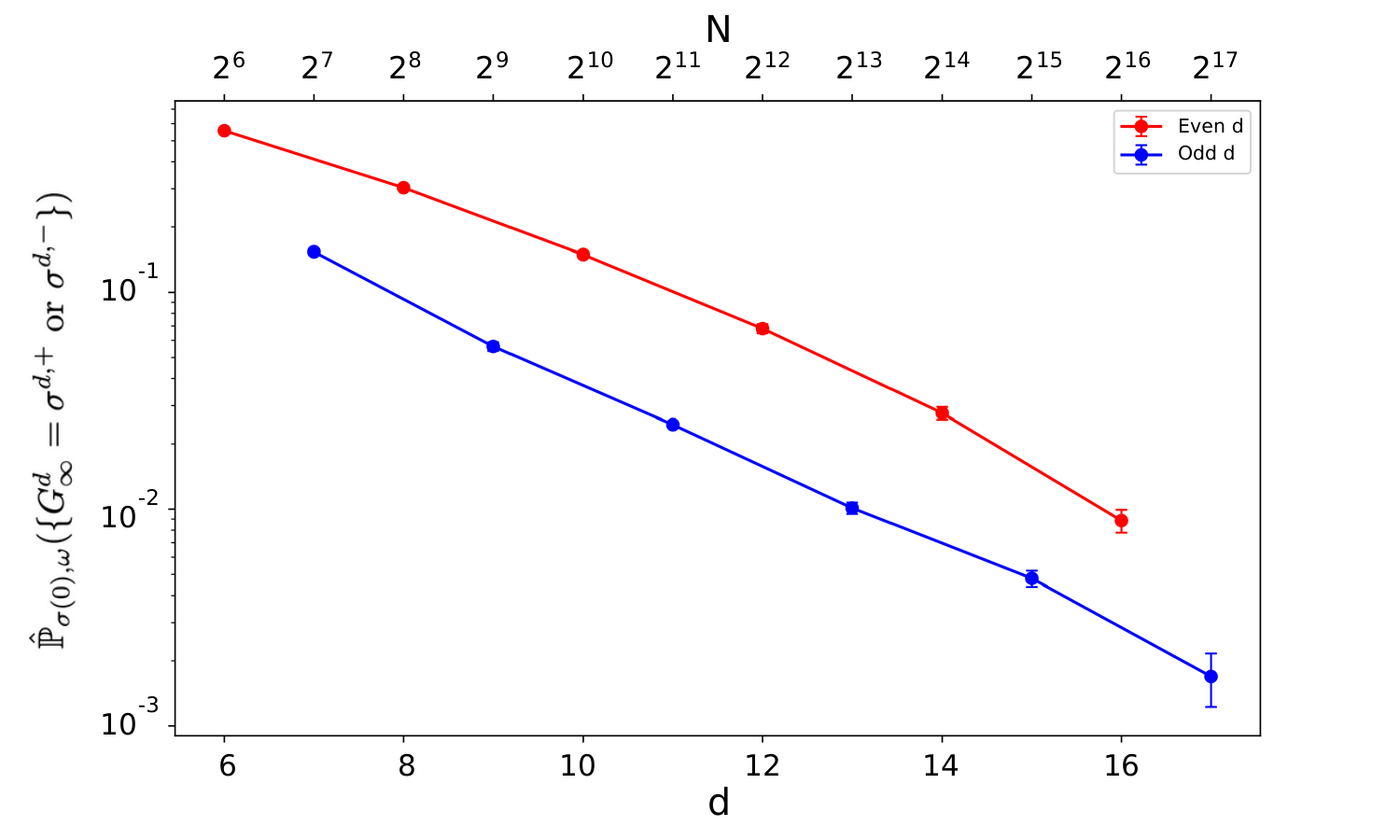}
    \caption{The probability to enter the ground state $\hat{\mathbb{P}}_{\sigma(0),\omega}(\{G^d_\infty = \sigma^{d,+}\text{ or }\sigma^{d,-}\})$ vs.\ dimension on a log-linear scale.}
    \label{fig:Figure2}
\end{figure}
\\\newline
Figures \ref{fig:Figure3}a and  \ref{fig:Figure3}b show the sample variance on log-log plots of $M_\infty(d)$ scaled by $\sqrt{N}$ and $N$. As $d\rightarrow\infty$, the difference between the sample variance in even and odd dimensions gradually decreases, suggesting the possible disappearance of the even-odd disparity as dimension increases. A typical central-limit-type behavior would have the scaling factor for the variance of $M_\infty(d)$ be $\sqrt{N}$, but \cref{fig:Figure3}a shows that the sample variance of $M_\infty(d)$ grows exponentially fast with dimension when scaled by $\sqrt{N}$. Alternatively, \cref{fig:Figure3}b shows that the sample variance of $M_\infty(d)$ decays with dimension when scaled by $N$. This suggests that the correct scaling factor should be somewhere between $\sqrt{N}$ and $N$. Numerically, an  $N/\log\log N$ scaling gives a good fit, as shown in \cref{fig:Figure3}c. But one should keep in mind that CLT presumes that for sufficiently large d, the density concentrates near zero, and that we are very far from that asymptotic behavior, so we cannot say much at this point about the asymptotic behavior.
\begin{figure}[ht!]
    \centering
    \begin{minipage}{0.4\textwidth}
        \centering
        \includegraphics[width=\linewidth]{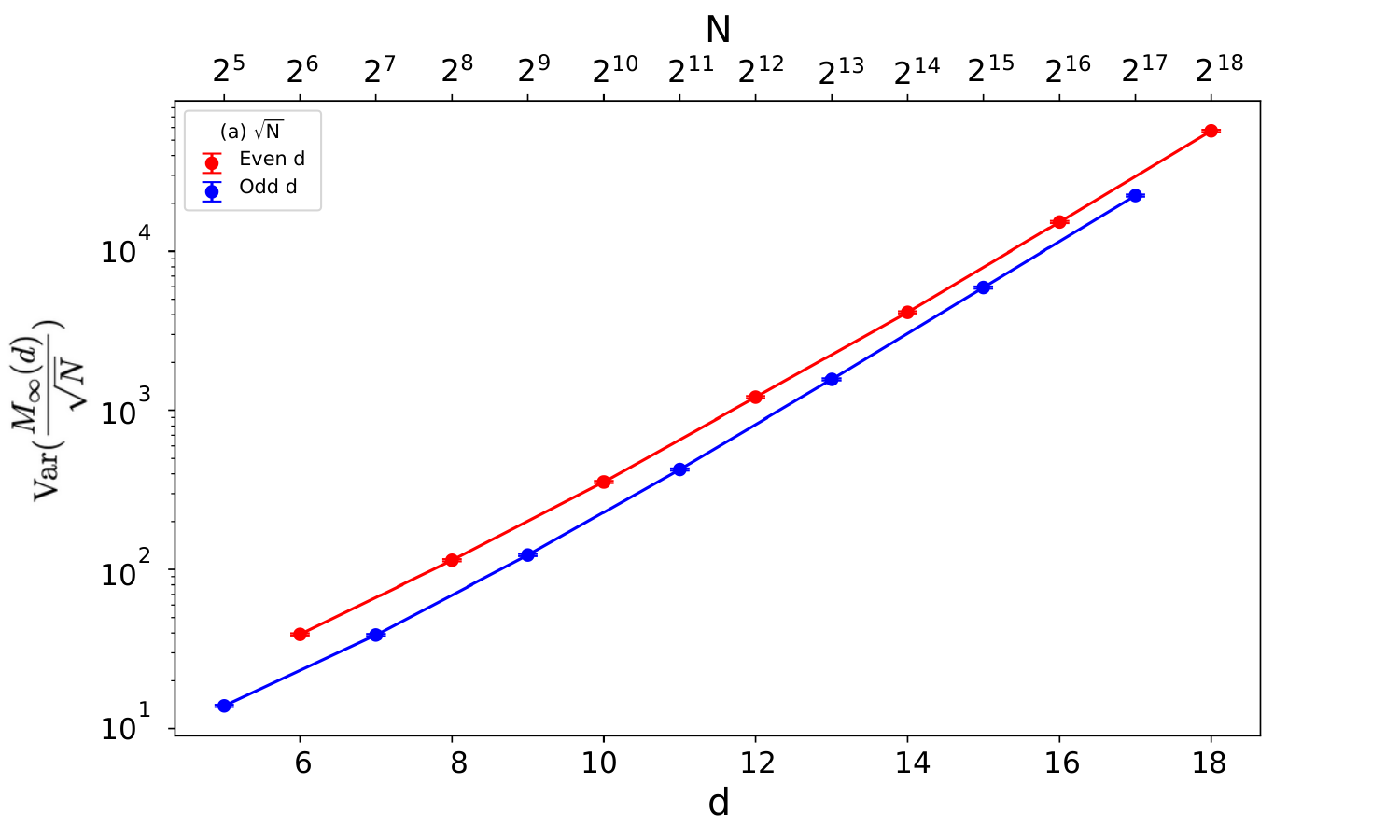}
        \caption*{(a)}
    \end{minipage}%
    \hfill
    \begin{minipage}{0.4\textwidth}
        \centering
        \includegraphics[width=\linewidth]{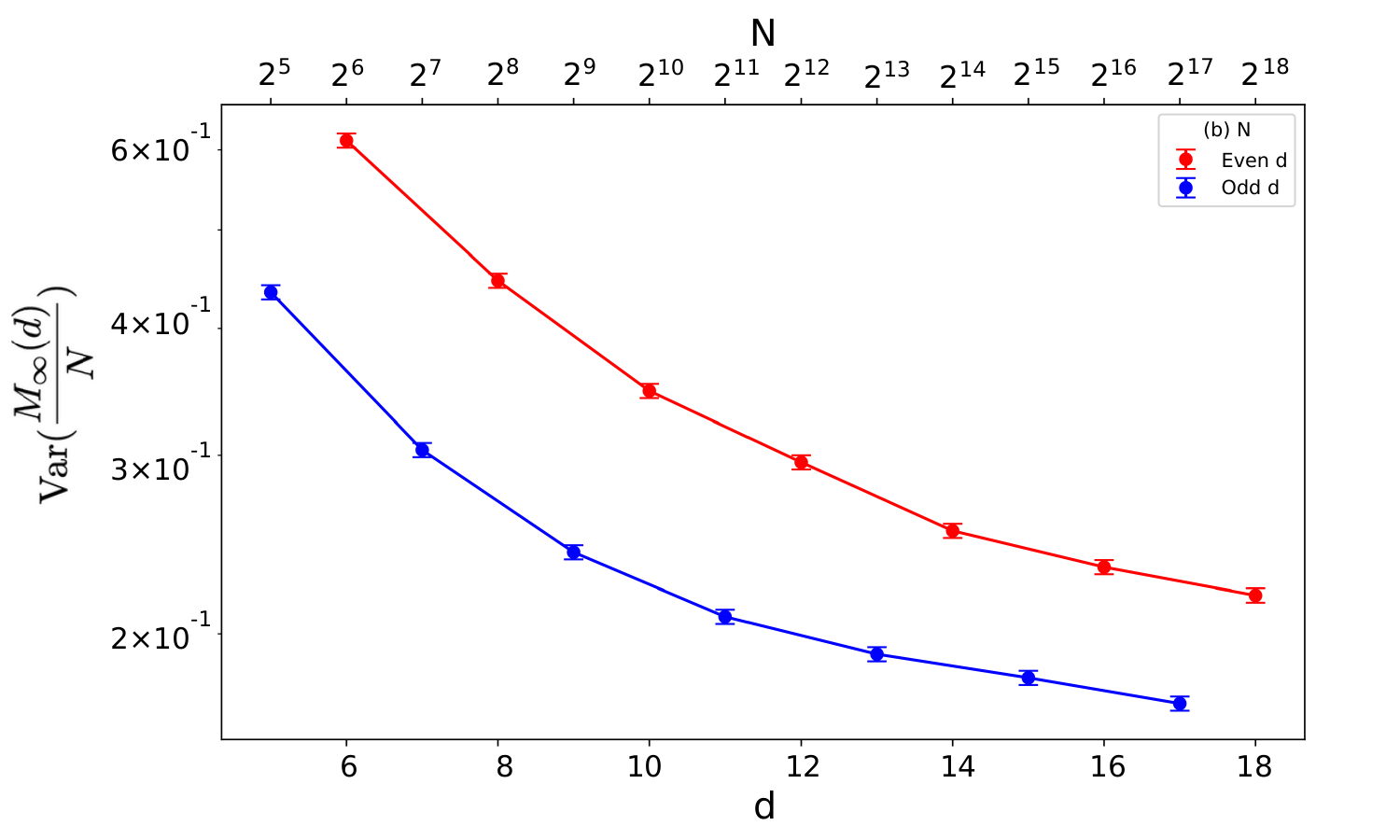}
        \caption*{(b)}
    \end{minipage}
    \begin{minipage}{0.4\textwidth}
        \centering
        \includegraphics[width=\linewidth]{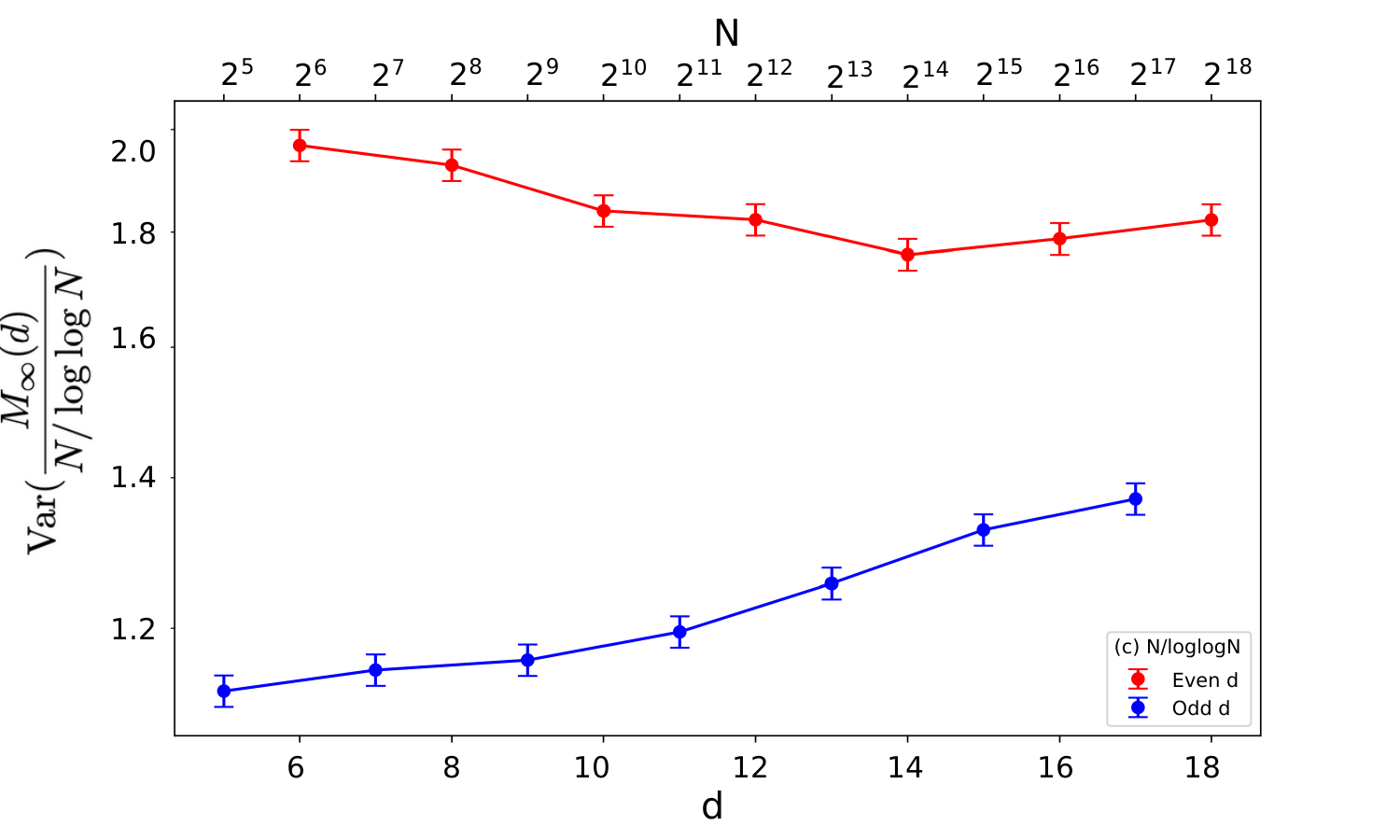}
        \caption*{(c)}
    \end{minipage}

    \caption{Sample variance of $M_\infty(d)$ under three different scalings at various dimensions: (a) scaled by $\sqrt{N}$, (b) scaled by $N$, and (c) scaled by $N / \log \log N$.}
    \label{fig:Figure3}
\end{figure}
\begin{figure}[h!]
    \centering
    \includegraphics[scale = 0.38]{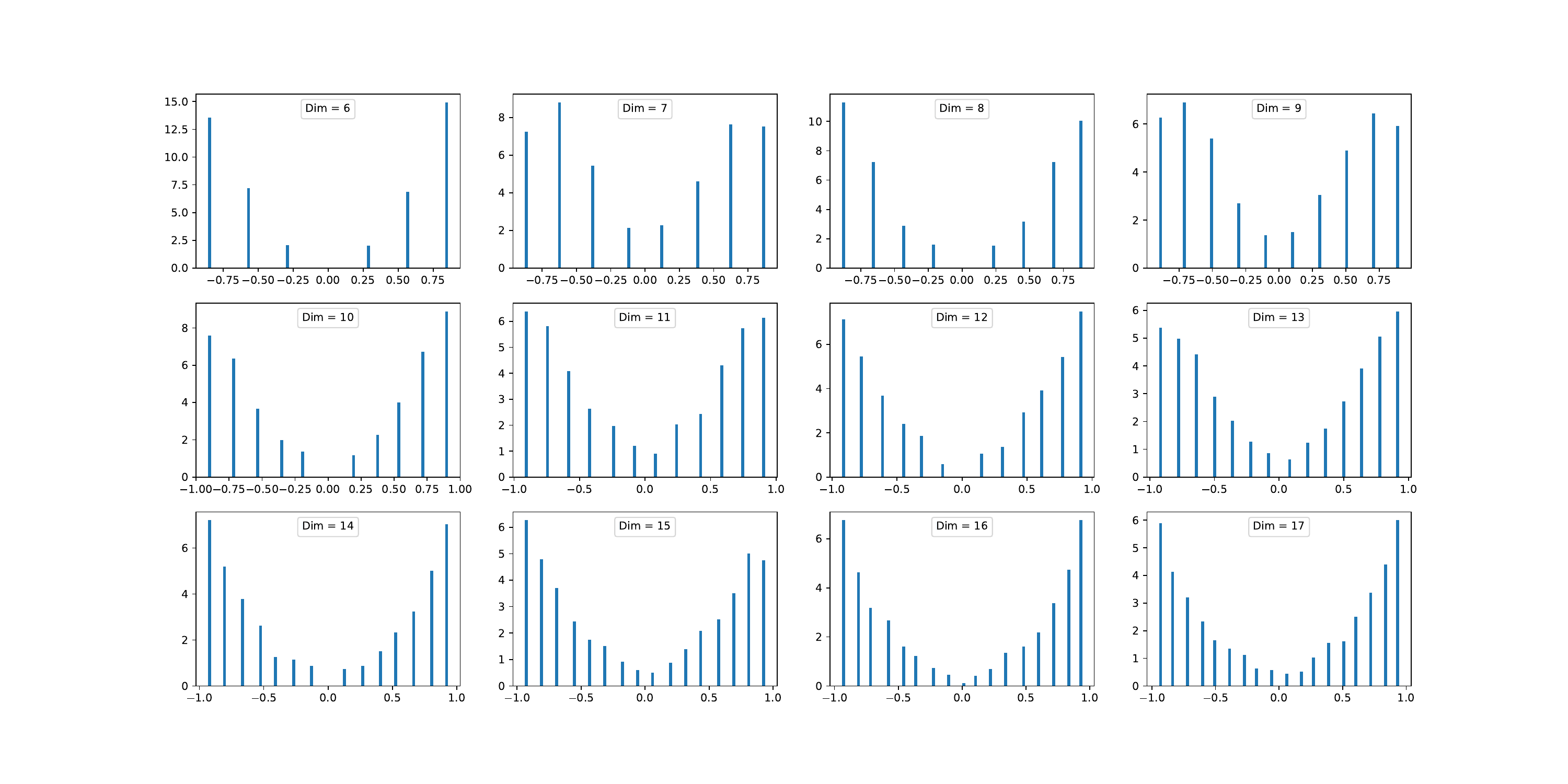}
    \caption{Sample distribution of $m_\infty(d)$ at different dimensions}
    \label{fig:Figure4}
\end{figure}\\
\noindent We next study the local fields (at the origin), $m_\infty (d)$, of the final configuration, see Eq.\ \eqref{eq:m}.
This carries information about the tendency of neighboring spins in the final state to align. \cref{fig:Figure4} shows the sample distribution of $m_\infty(d)$ for different dimensions. As $d$ increases, the distribution neither converges to a delta function at $0$ (corresponding to zero-correlation between neighboring spins) nor to delta functions at $+1$ and $-1$ (all neighboring spins are aligned). Instead, the distribution resembles a concave-up parabola with a minimum at $0$. Note that in even $d$, the probability mass at $0$ corresponds to $\mathbb{P}_{\sigma(0),\omega}(B^d_i)$, the probability for a single site to become a blinker as $t\rightarrow\infty$.\\\newline
Another feature that is closely related to the local fields is the final energy per spin $\mathcal{H}_\infty(d)$. The final energy of a spin $\sigma_i$ equals $-d\sigma_im^{i}_\infty(d)/2$, so the mean final energy per spin $\mathcal{H}_\infty(d)$ is given by $d/2$ times the averge of $|m_\infty(d)|$.  Figure \ref{fig:Figure4} suggests the probability density of $m_\infty(d)$ converges to an even quadratic that vanishes at zero, in which case the mean energy per site should converge for large $d$ to, $\mathcal{H}_\infty(d)=-3d/8$.    \cref{fig:Figure5} shows the sample distribution of  $\mathcal{H}_\infty(d)$ for different values of $d$. The distribution is broad and skewed to the left. The last bin corresponds to configurations reaching the ground state, with energy per spin, $-d/2$. There is a relatively sharp cut-off on the right-hand side.
\begin{figure}[h!]
    \centering
    \includegraphics[scale = 0.38]{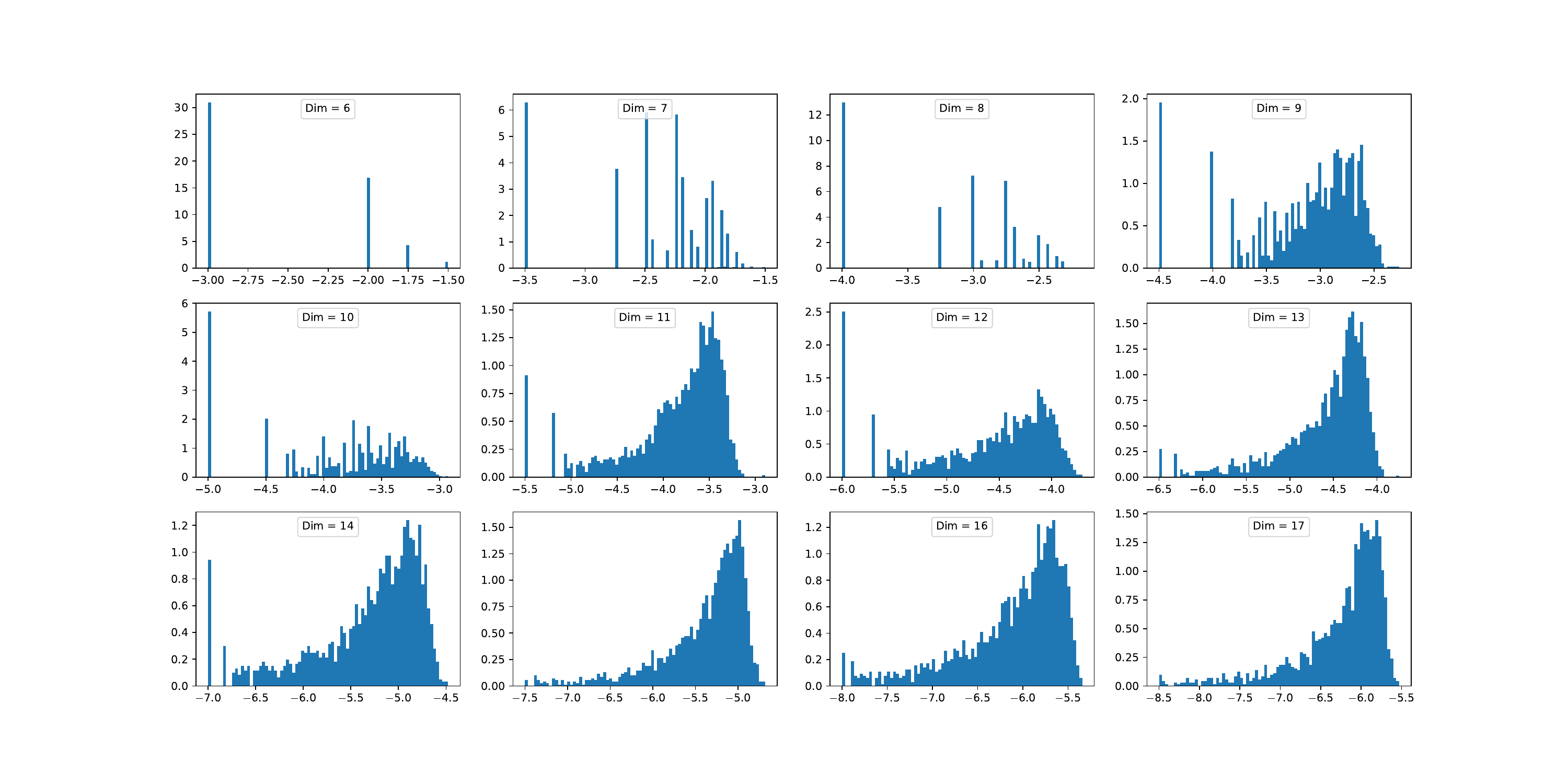}
    \caption{Sample distribution of average $\mathcal{H}_\infty(d)$ per site for different dimensions}
    \label{fig:Figure5}
\end{figure}
\\\newline
Finally, \cref{fig:Figure6} shows $\overline{f(d)}$, the mean number of flips per site vs.~dimension (omitting blinkers). It's worth noting that $\overline{f(d)}$ is small and grows slowly in both even and odd dimensions, suggesting that most spins on the hypercube require at most a single flip (with many having none at all) to reach their final states. $\overline{f(d)}$ in even dimensions is larger than in odd dimensions; this is because of the existence of `ties' (i.e, $\Delta E_i = 0$) in even dimensions. It then makes sense that it generally requires a greater number of flips for even-dimensional hypercubes to enter a final state.
\begin{figure}[h!]
    \centering
    \includegraphics[scale=0.40]{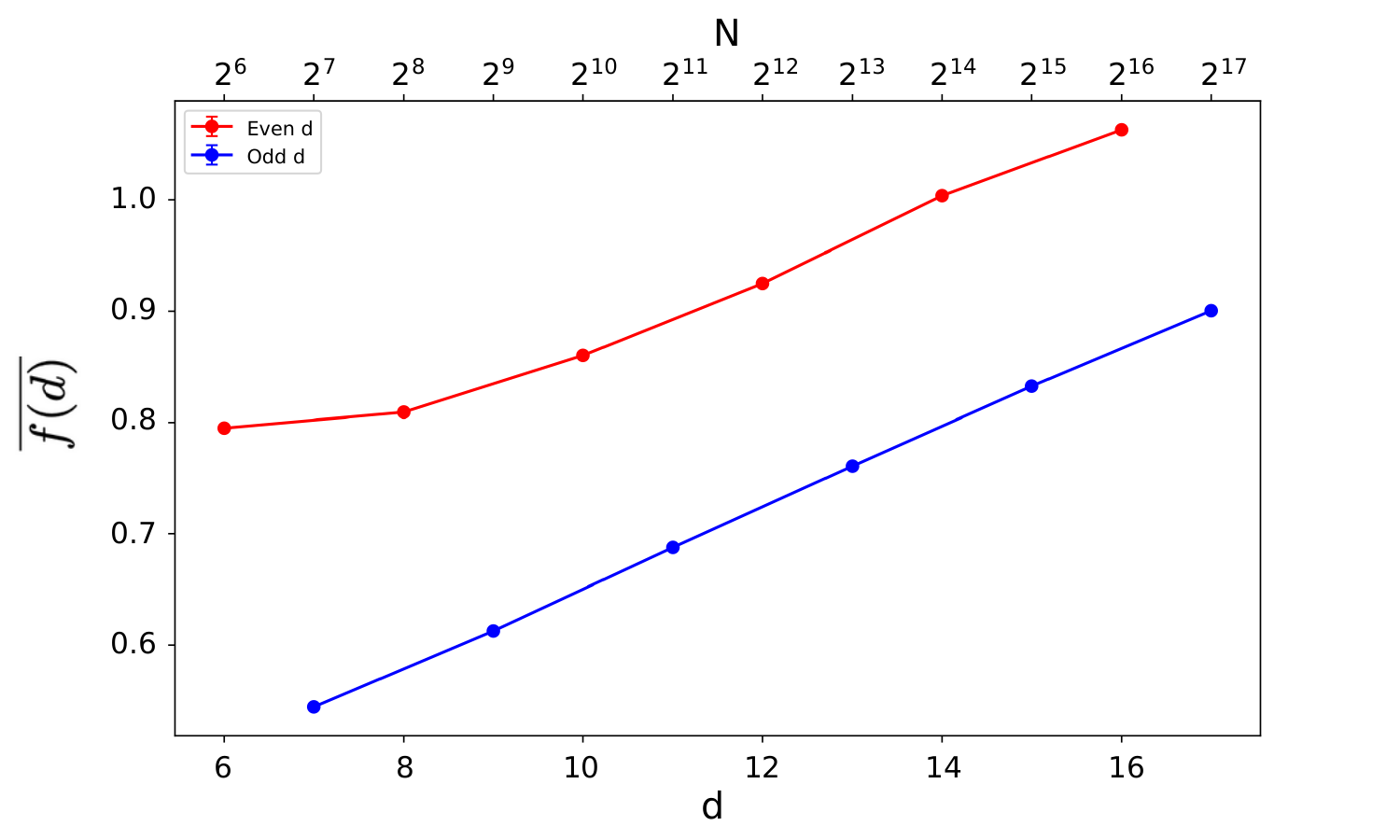}
    \caption{$\overline{f(d)}$ at different dimensions}
    \label{fig:Figure6}
\end{figure}\\\newline
Many substructures in $d$-dimensional hypercubes can be frozen if fully magnetized (i.e., all spins are of the same sign): these are the $k$-dimensional subcubes, unions of $k$-dimensional subcubes (with $k\geq \lfloor \frac{d}{2}\rfloor+1$), and more generally, $k$-cores. During the dynamics, once such a structure is fully magnetized its size is nondecreasing. These various structures lead to many frozen states on hypercubes. As $d\rightarrow\infty$, their number increases rapidly  (roughly on the order of $2^{\sqrt{N}}$, see Section 4.1 and \cref{thm:2}). The reason behind   $\lim_{d\rightarrow\infty}\mathbb{P}_{\sigma(0),\omega}(\{G^d_\infty = \sigma^{d,+}\text{ or }\sigma^{d,-}\}) = 0$ is that the hypercube becomes increasingly likely to get trapped in one of the many frozen states, preventing it from reaching a ground state. It is therefore interesting to ask: What is the typical largest cluster a frozen state can have? 
\begin{figure}[h!]
    \centering
    \includegraphics[scale=0.32]{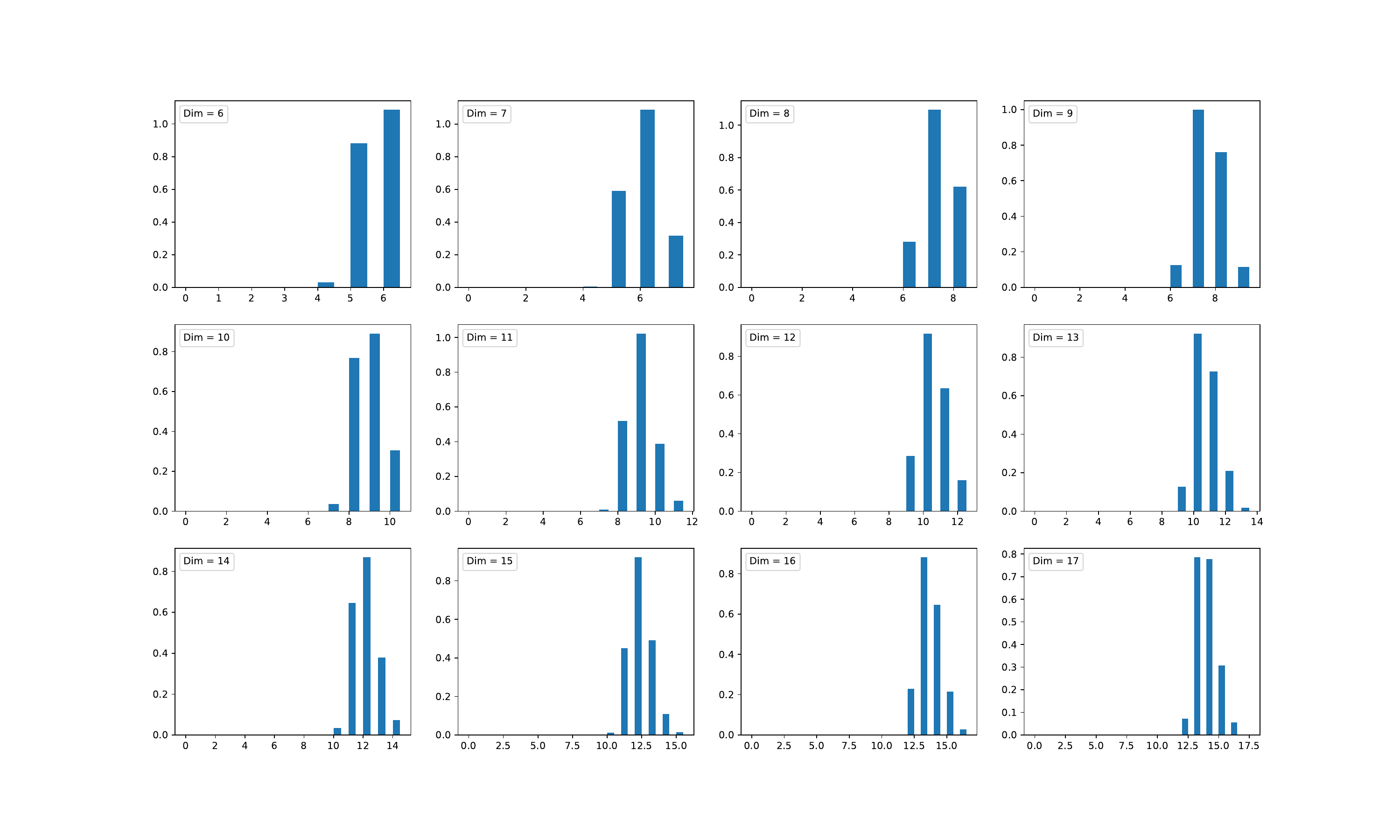}
    \caption{Sample distribution of $\kappa(d)$ at different dimensions}
    \label{fig:Figure7}
\end{figure}\\\newline
Figure \ref{fig:Figure7} shows the sample distribution of the core number of the largest $k$-core,~$\kappa(d)$ at different dimensions;  note that $\kappa(d) = d$ is equivalent to the system reaching a ground state. The distribution of $\kappa(d)$ lies in the interval 
$(\lfloor \frac{d}{2} \rfloor + 1,d)$
, because no $k$-core with $k<\lfloor \frac{d}{2} \rfloor + 1$ is dynamically stable. As $d$ increases, the distribution concentrates around $3d/4$, suggesting that a typical largest cluster in $d$ dimensions has core number close to $3d/4$. The tails of the distribution on either side fall to $0$, corresponding on the right to $\lim_{d\rightarrow\infty}\mathbb{P}_{\sigma(0),\omega}(\{G^d_\infty = \sigma^{d,+}\text{ or }\sigma^{d,-}\}) = 0$.
\begin{figure}[h!]
    \centering
        \includegraphics[scale=0.40]{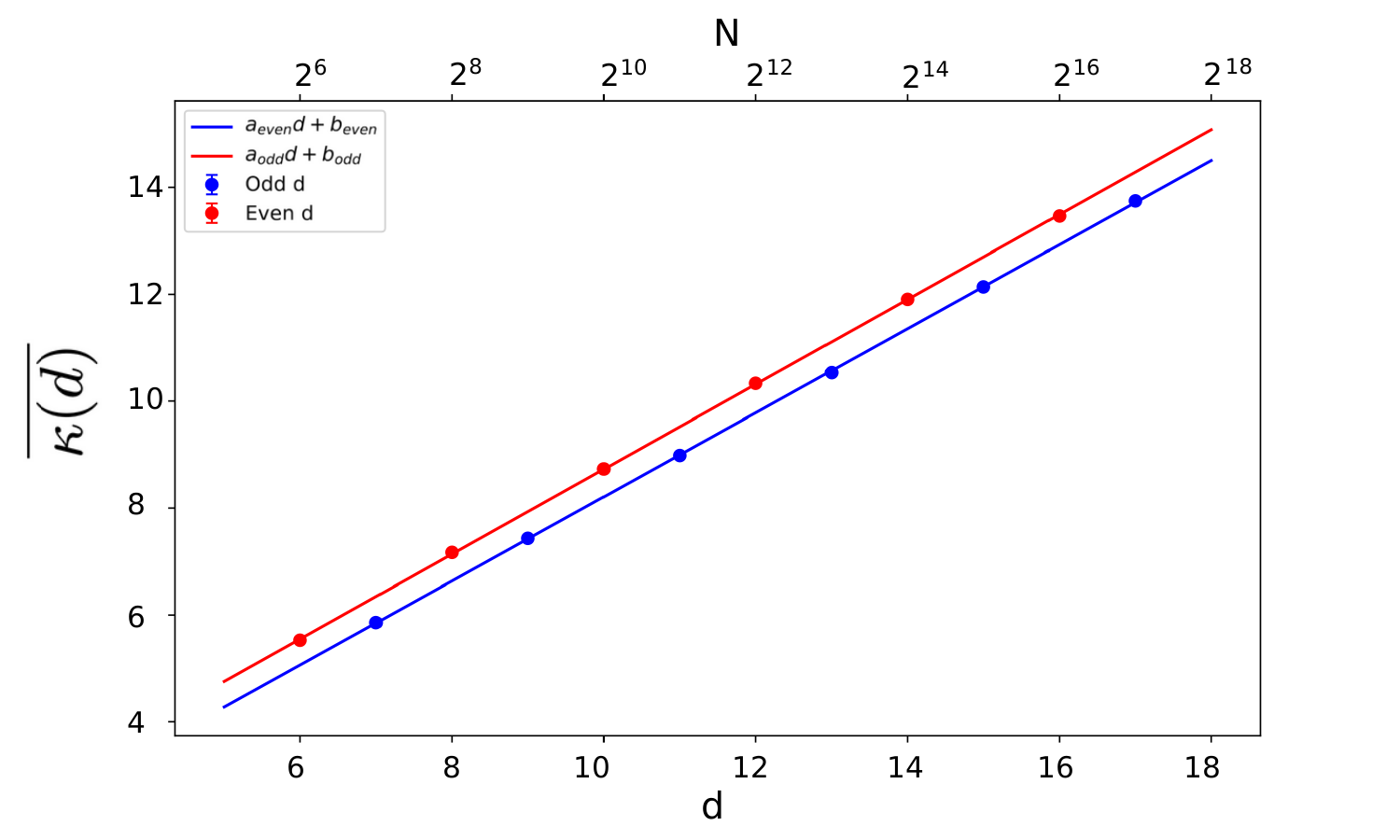}
    \caption{ Mean of $\kappa(d)$ at different dimensions with linear fits of the form $\kappa(d)=a d+ b$  and fitting parameters: $a_{even} = 0.731\pm0.001, b_{even} = 0.710\pm0.023$, $a_{odd} = 0.733\pm0.001, b_{odd} = 0.167\pm0.024$}
        \label{fig:Figure8}
\end{figure} \\\newline
\noindent We further study the sample average of the largest $k$-core in different dimensions;  \cref{fig:Figure8} shows the data together with a linear fit. Note that the average of the largest $k$-core jumps between odd and even dimensions, due to the floor function in the dynamical stability requirement $k\geq \lfloor\frac{d}{2}\rfloor + 1$. The slope parameter in both even and odd dimensions is close to $3/4$, corresponding to the concentration of the largest k-core around $3d/4$.\\\newline
We finally turn our study to blinkers, spins that flip infinitely often as $t\rightarrow\infty$, which, as noted earlier, can occur only in even dimensions. We look first at  $\mathbb{P}_{\sigma(0),\omega}(\sigma^d_i \in g^d_\infty)$, the probability for the spin $\sigma_i$ to eventually become a blinker. Figure \ref{fig:Figure9} shows its estimator $\hat{\mathbb{P}}_{\sigma(0),\omega}(\sigma^d_i \in g^d_\infty)$, which is the sample fraction of the number of blinkers in different dimensions (in units of $10^{-5}$). We see from the plot that $\hat{\mathbb{P}}_{\sigma(0),\omega}(\sigma^d_i \in g^d_\infty)$ grows slowly with dimension. Though we believe the growth is transient and it must go down to zero if the density of blinkers goes to zero, as we conjecture later. We are not yet in asymptopia because the probability of having a blinker is growing more rapidly than the conditional density is shrinking.
\begin{figure}[h!]
    \centering
    \includegraphics[scale=0.41]{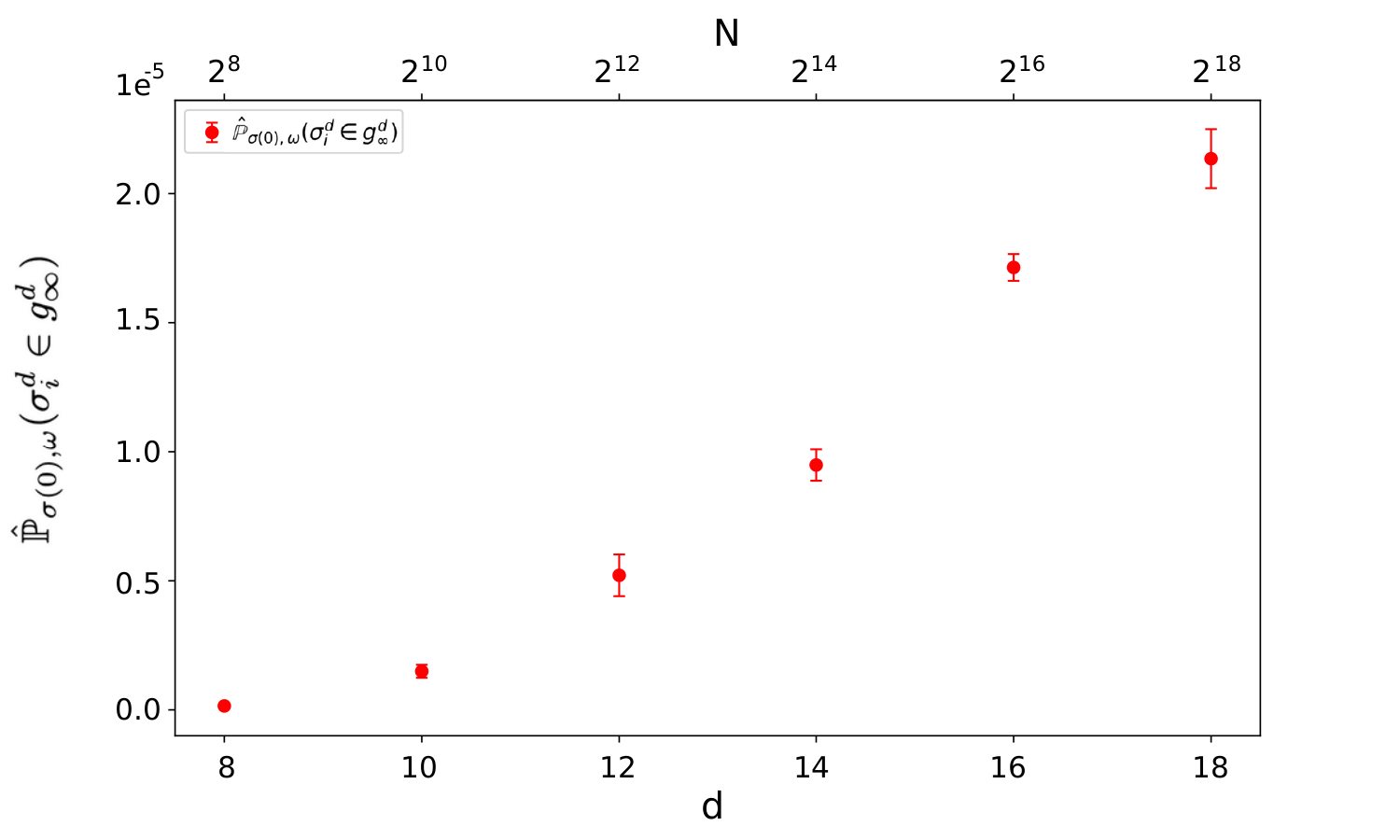}
    \caption{The probability that a spin is a blinker, $\hat{\mathbb{P}}_{\sigma(0),\omega}(\sigma^d_i \in g^d_\infty)$, at different dimensions}
    \label{fig:Figure9}
\end{figure}\\\newline
We have not observed any blinkers in dimensions $d<8$, and in Section 4 we will show that eight is in fact the lowest dimension in which blinkers can appear. The intuition behind the theoretical argument in the next section is that the existence of a blinker places certain requirements on its surrounding environment. One consequence of this is that there is an effective `repulsion' between individual blinkers, so that in lower dimensions there may only be sufficient ``room'' for a single blinker.  It is therefore useful to consider the conditional probability $\mathbb{P}_{\sigma(0),\omega}(\sigma^d_i \in g^d_\infty||G^d_\infty| > 1)$ for the spin $\sigma_i$ to eventually become a blinker conditioned on the event that there is at least one other blinker in the hypercube. \cref{fig:Figure10} shows its estimator $\hat{\mathbb{P}}_{\sigma(0),\omega}(\sigma^d_i \in g^d_\infty||G^d_\infty| > 1)$, which is the averaged fraction of the number of blinkers in different dimensions, conditioned on discarding all runs where no blinkers appear. For this estimator we find a fit of $ad^{-1}2^{-\frac{d}{2}}$. The decaying fraction of blinkers as dimension increases is an indication that the size of the 'surrounding environment' needed for an isolated blinker is growing with dimension. 
\begin{figure}[h!]
    \centering
    \includegraphics[scale=0.41]{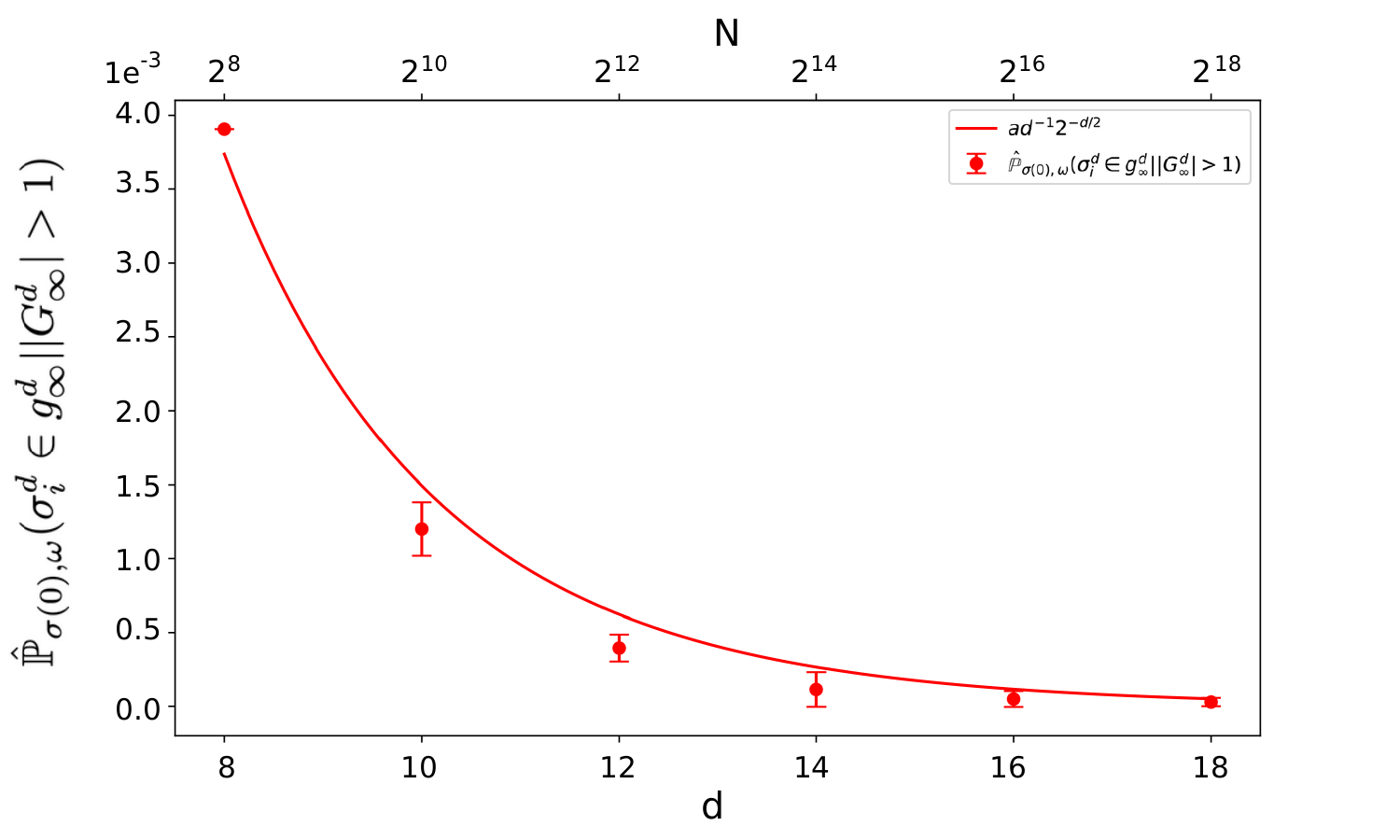}
    \caption{$\hat{\mathbb{P}}_{\sigma(0),\omega}(\sigma^d_i \in g^d_\infty||G^d_\infty| > 1)$ at different dimensions, fitting parameters: $a = 2.09\pm0.10$}
    \label{fig:Figure10}
\end{figure}\\\newline
Finally, we consider $\mathbb{P}_{\sigma(0),\omega}(|G^d_\infty| > 1)$, the probability for a hypercube to contain at least one blinker. Figure \ref{fig:Figure11} shows its estimator $\hat{\mathbb{P}}_{\sigma(0),\omega}(|G^d_\infty| > 1)$, the frequency of hypercubes that contain at least one blinker. We see that the frequency is growing rapidly with dimension, and conjecture that it approaches one as $d\rightarrow \infty$. Consequently, even-dimensional hypercubes with large $d$ should contain at least one blinker. But given the decay of $\mathbb{P}_{\sigma(0),\omega}(\sigma^d_i \in g^d_\infty||G^d_\infty| > 1)$, the density of blinkers will be small in large dimensions.
\begin{figure}[h!]
    \centering
    \includegraphics[scale=0.42]{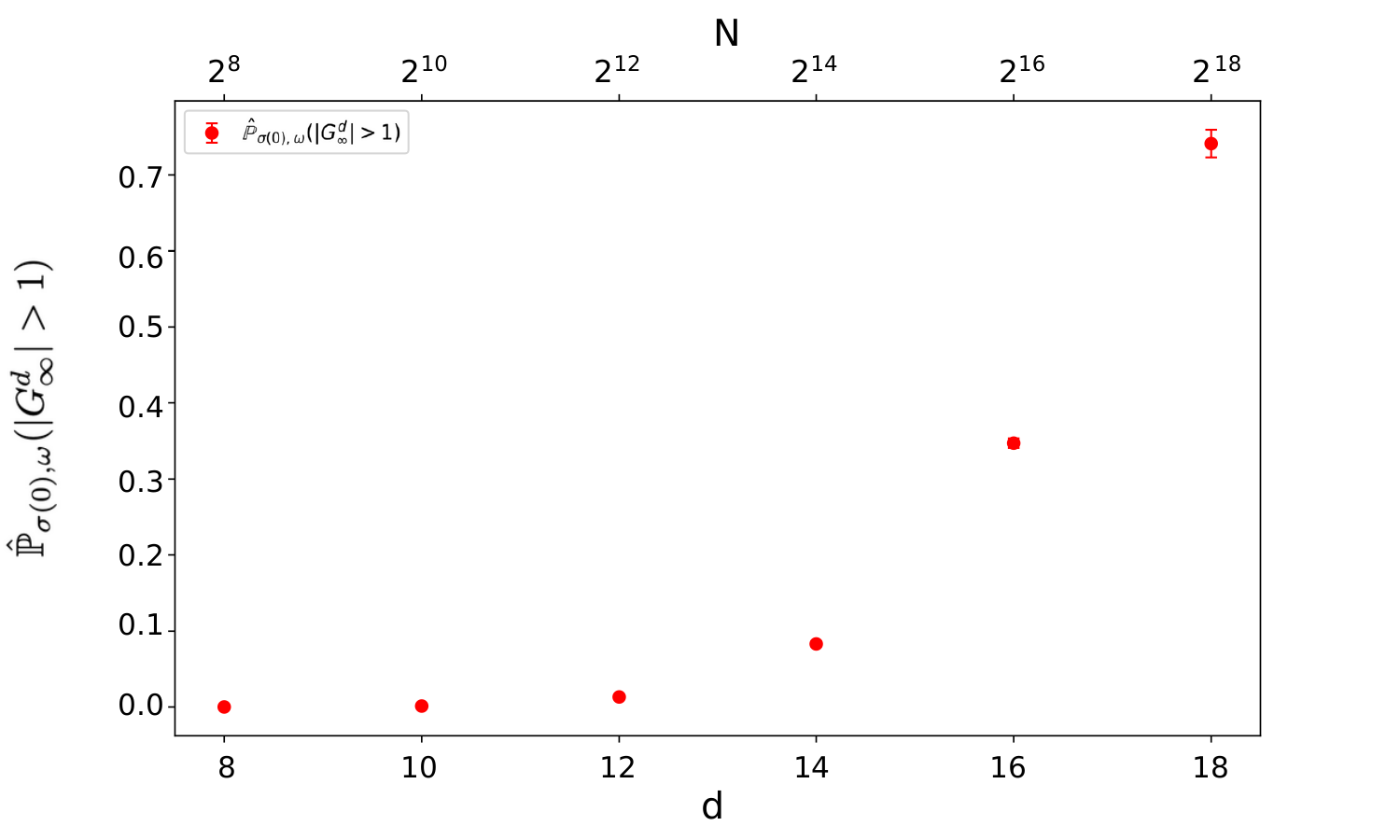}
    \caption{$\hat{\mathbb{P}}_{\sigma(0),\omega}(|G^d_\infty| > 1)$ at different dimensions}
    \label{fig:Figure11}
\end{figure}
\subsection{Double-copy Dynamics}
\noindent We next study $q_\infty(d)$, the final overlap between `identical twins' starting with the same initial configuration but evolving independently. Figure \ref{fig:Figure12} shows the sample distribution of $q_\infty(d)$ at different dimensions.
Similar to our results on the final magnetization, the sample distribution of $q_\infty(d)$ is generally broader in even dimensions. The distribution of $q_\infty(d)$ generally skews to the right, and neither concentrates around $1$ (which is the case for the diluted Curie-Weiss Model\cite{Gheissari2018}), nor concentrates around $0$; in fact, $q_\infty(d)$ appears to be concentrating around $0.5$ as $d$ grows.
\begin{figure}[h!]
    \centering
    \includegraphics[scale=0.40]{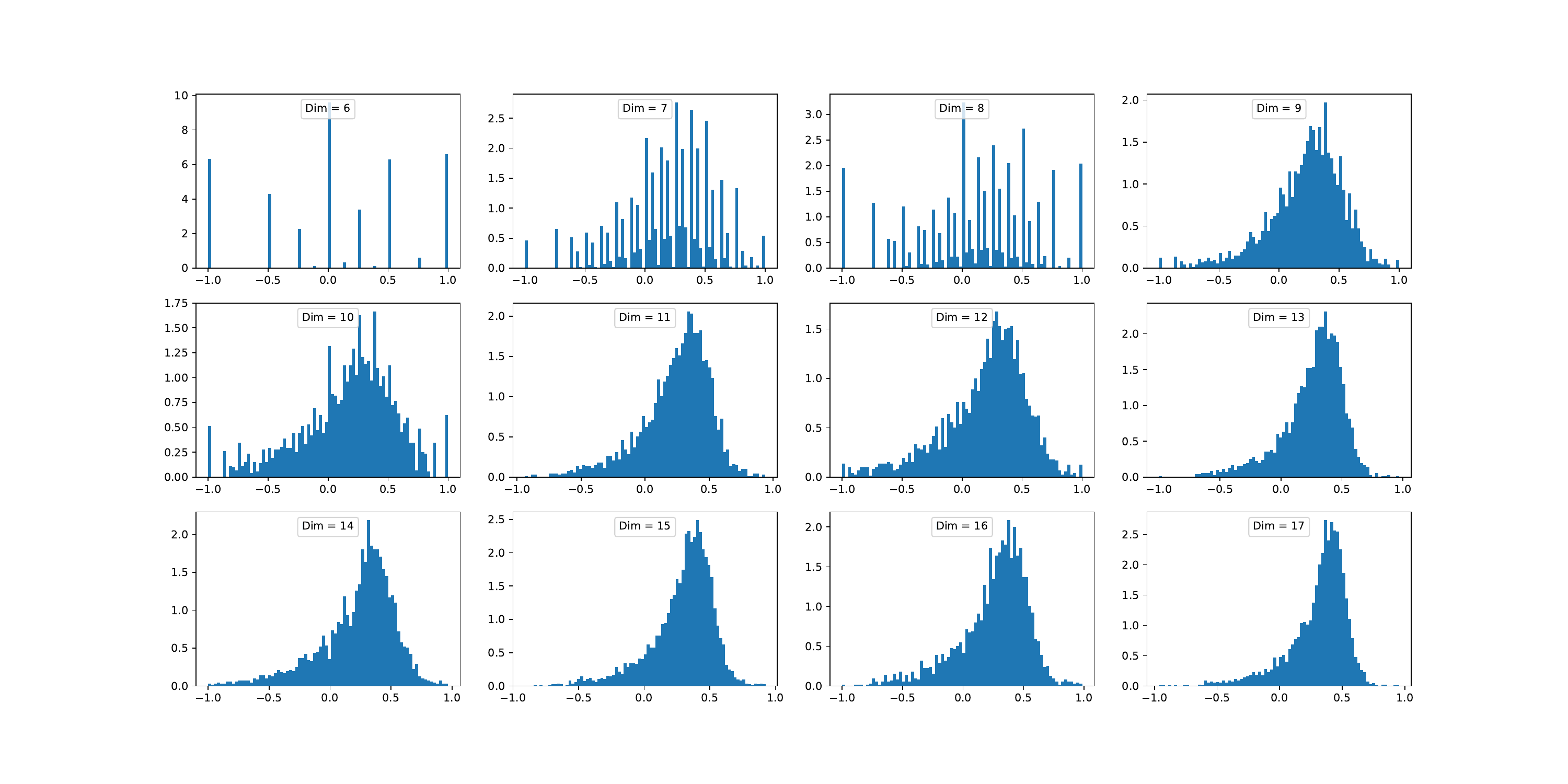}
    \caption{Sample distribution of the final overlap, $q_\infty(d)$, at different dimensions}
    \label{fig:Figure12}
\end{figure}\\\newline
To study the limit of $q_\infty(d)$ as $d\rightarrow\infty$, we consider its sample mean $\overline{q_\infty(d)}$ as a function of $d$. Again $\overline{q_\infty(d)}$ alternates between even and odd dimensions but is increasing with $d$ in both. \cref{fig:Figure13} show a fit for $\overline{q_\infty(d)}$ using a simple power law $c-ad^b$, where $0<c<1$. Clearly, there should be two power laws $c-a_{even}d^b$, $c-a_{odd}d^b$ to fit the even-odd difference, but at the same time, we should restrict them with the same asymptote $c$ as we believe such differences should vanish for large $d$. The fit suggests that $\overline{q_\infty(d)}\rightarrow 0.58\pm 0.14$, corresponding to concentration of $q_\infty(d)$ near $0.5$ as seen previously in \cref{fig:Figure12}.
\\
\begin{figure}[h]
    \centering
        \includegraphics[scale=0.42]{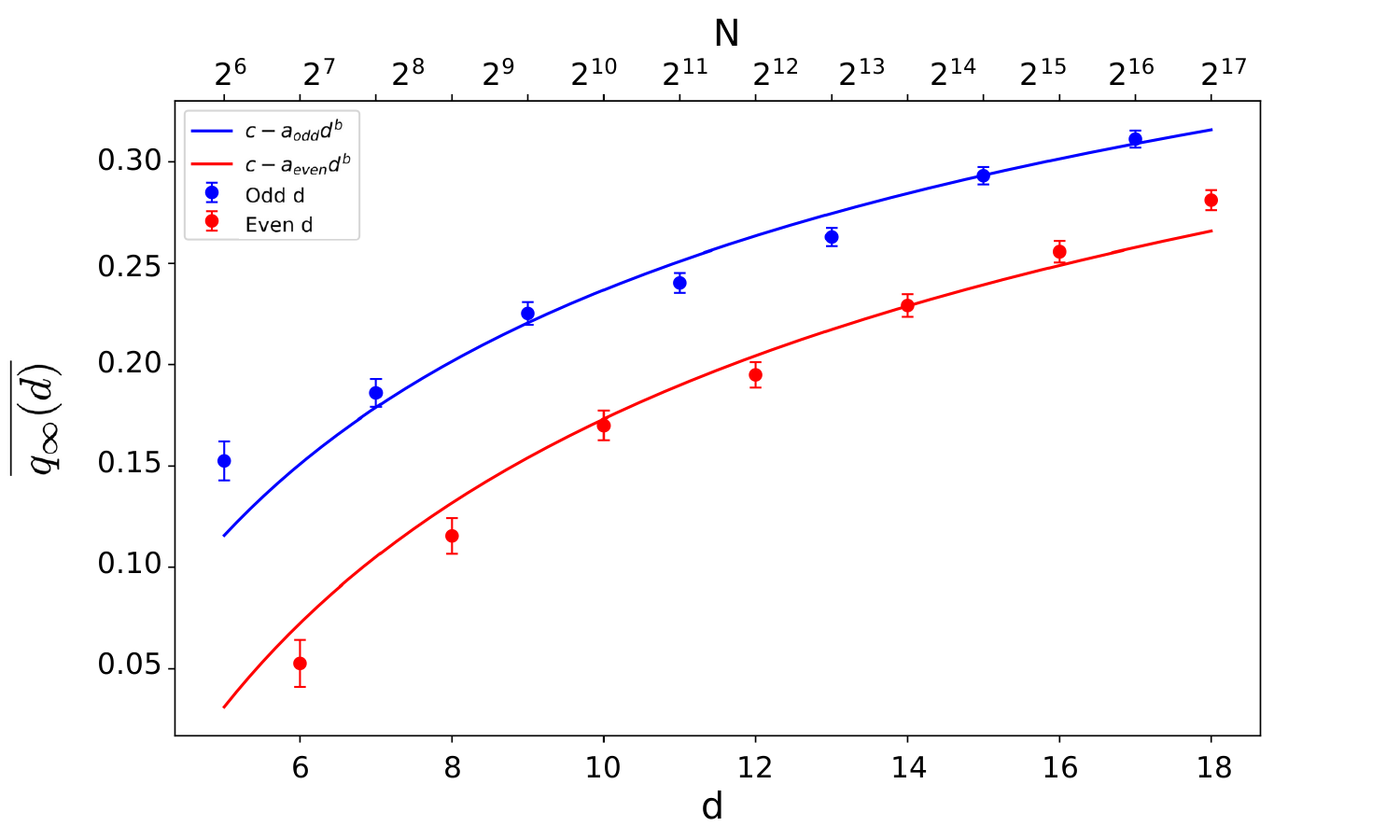}
        \caption{Sample mean $\overline{q_\infty(d)}$ at different dimensions, fitted with same power parameter $b$ in even and odd dimensions: $a_{even} = 1.13\pm0.09, a_{odd} = 0.95\pm0.04, b = -0.44\pm0.18, c=0.58\pm0.14$}
        \label{fig:Figure13}
    
\end{figure}
\section{Theoretical Results and Constructions}

\subsection{Frozen States}
\noindent We begin this section by introducing a renormalization scheme that exploits the recursive structure of the hypercube. 
\begin{definition}
A \textbf{k-renormalization} on the hypercube $Q_d$ is defined as the mapping of $Q_d$ onto a lower-dimensional hypercube $\tilde{Q}_{d-k}$, where:
\begin{enumerate}
    \item Each vertex $i$ of $\tilde{Q}_{d-k}$ corresponds to a $k$-dimensional sub-cube $Q_{k,i}$ of $Q_d$.
    \item Each edge between two vertices $i$ and $j$ in $\tilde{Q}_{d-k}$ corresponds to the $k$ edges connecting the vertices of $Q_{k,i}$ and $Q_{k,j}$. Specifically, these edges connect every site in $Q_{k,i}$ to the corresponding site in $Q_{k,j}$.
\end{enumerate}
\end{definition}
\noindent Remark: Such a renormalization is unique only up to a permutation. One can fix an arbitrary set of $d-k$ coordinates to be the vertices $j$, freeing the remaining $k$ coordinates to constitute the inside of each sub-cube $Q_{k,j}$. Hence one can have a number ${n \choose k}$ of such renormalizations.\\

\noindent We first present a construction that allows one to combine two arbitrary $d$-dimensional stable configurations into a single $(d+2)$-dimensional stable configuration.
\newline
\setcounter{theorem}{0} 
\begin{theorem}\label{thm:1}
Given any two frozen states $\sigma^{d,a},\sigma^{d,b}\in \{-1,1\}^{Q_d}$ for a $d$-dimensional ($d\geq3$) hypercube, a $2$-renormalization leads to the following two $(d+2)$-dimensional  configurations which are also frozen states --- see the figure:
\newpage

\begin{equation}
\sigma^{d+2,1}_{(i, i_{d+1}, i_{d+2})} = 
\begin{cases} 
      \sigma^{d,-}_{i} & \text{if } (i_{d+1}, i_{d+2}) = (0, 0), \\[8pt]
      \sigma^{d,a}_{i} & \text{if } (i_{d+1}, i_{d+2}) = (0, 1), \\[8pt]
      \sigma^{d,b}_{i} & \text{if } (i_{d+1}, i_{d+2}) = (1, 0), \\[8pt]
      \sigma^{d,+}_{i} & \text{if } (i_{d+1}, i_{d+2}) = (1, 1),
\end{cases}
\end{equation}

where $i\in Q_d$ and $(i_{d+1}, i_{d+2})\in\{0,1\}^2$.
\end{theorem}
\begin{figure}[h]
    \centering

    \includegraphics[width=0.4\textwidth,trim=0cm 19cm 10cm 2.1cm, clip]{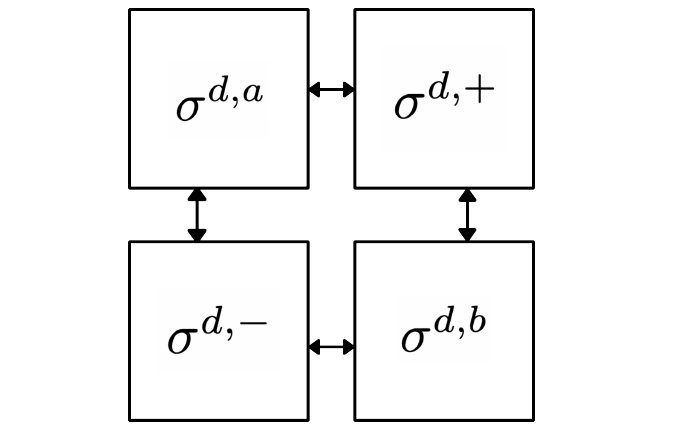}    
    \caption{Graphs for constructions in \cref{thm:1}; each square represents a d-dim hypercube, a double-headed arrow between two hypercubes means connections between all spins with the same (d-dimensional) coordinates.}
\end{figure}
\begin{proof}
We consider the energy change $\Delta E^{d+2}$ for each spin in $\sigma^{d+2,1}$. For $i\in Q_k$, if $i' = (i,0,0)$, we have $\Delta E^{d+2}_{i'} = \Delta E^d_i > 0$, since $\sigma^{d+2,1}_{i'}$ is only additionally connected to a plus spin and a minus spin compared to $\sigma^{d,a}_{i}$. Similarly, $\Delta E^{d+2}_{i'} > 0$ if $i' = (i,1,1)$. Now if $i' = (i,1,0)\text{ or }(i,0,1)$, $\Delta E^{d+2}_{i'} \geq \Delta E^d_i-2 > 0$. Hence $\sigma^{d+2,1}$ is a final state, and $\sigma^{d+2,2}$ can similarly be shown to be a final state.
\end{proof}
\noindent Now suppose that the number of $d$-dimensional frozen states is $n$. By \cref{thm:1}, there are $n$ possible choices for each of $\sigma^{d,a}$ and $\sigma^{d,b}$, so  the number of $(d+2)$-dimensional frozen states is at least $n^2$. Hence one should expect the number of frozen states to grow at least as fast as $2^{2^{d/2}} = 2^{\sqrt{N}}$. The simplest way to construct a frozen state on a $d$-dimensional hypercube $Q_d$ is as follows:
\setcounter{example}{0} 
\begin{example}\label{exp:1}
For $k\geq\lfloor \frac{d}{2}\rfloor+1$, consider the $k$-renormalization of $Q_d$ into $\tilde{Q}_{d-k}$, so that each vertex $j$ of $\tilde{Q}_{d-k}$ is a $k$-dimensional sub-cube $Q_{k,j}$. Fix $j_0\in\{0,1\}^{\tilde{Q}_{d-k}}$, set all spins in $Q_{k,j_0}$ to be plus, and all spins in $Q_{k,j_0}^c$ to be minus. The resulting configuration is a frozen state.
\end{example}
\begin{proof}
Each spin in $Q_{k,j_0}$ has $k$ plus neighbors and $d-k$ minus neighbors, hence is stable. There are $d$ sub-cubes $Q_{k,j}$ that are connected to $Q_{k,j_0}$, where the Hamming distance $d_h(j,j_0) = 1$. For each such $j$, all spins in $Q_{k,j}$ have one plus neighbor and $d-1$ minus neighbors, hence is also stable. Since the rest of the spins have only minus neighbors, this configuration is a frozen state.
\end{proof}
\noindent We now provide a lower bound for the number of frozen states, in a manner similar to \cref{exp:1}. 
\begin{theorem}\label{thm:2}
Let $\mathcal{F}^d$ denote the set of all frozen spin configurations $\{\sigma^d\text{ frozen}\}$ for $Q_d$. Then  $|\mathcal{F}^d|\succeq2^{\sqrt{N}}\, .$
\end{theorem}
\begin{proof}
For the sake of simplicity, assume $d$ even. We first $(\frac{d}{2}+1)$-renormalize $Q_d$ into $\Tilde{Q}_{\frac{d}{2}-1}$, so that each vertex $j$ of $\tilde{Q}_{\frac{d}{2}-1}$ is a ($\frac{d}{2}+1$)-dimensional sub-cube $Q_{\frac{d}{2}+1,j}$. For each $j$, set the spins in $Q_{\frac{d}{2}+1,j}$ to be either all plus or all minus. Then each spin in $Q_{\frac{d}{2}+1,j}$ will have at least $\frac{d}{2}+1$ neighbors that agree with it, and it is therefore frozen. In total there are $2^{\frac{d}{2}-1}$ such sub-cubes, hence there are at least $2^{2^{\frac{d}{2}-1}}$ frozen configurations.
\end{proof}
\noindent But the construction of \cref{thm:2} does not yield all possible frozen configurations. Any structure that is a $k$-core is frozen, hence as long as $k$ satisfies the dimension requirement $k\geq\lfloor \frac{d}{2}\rfloor+1$, in principle any configuration in which the set of plus spins (and similarly minus spins) is a union (not necessarily disjoint) of sub-cubes $Q_k$ would be a frozen configuration. One may then ask, do there exist frozen configurations that differ from those of \cref{thm:2}? The answer is yes, and such frozen configurations can occur in surprisingly low dimensions. 
\setcounter{example}{1} 
\begin{example}\label{exp:2}
On a $5$-dim hypercube $Q_5$, set all the spins on the sites in \cref{fig:Figure16} to be plus, and set all the remaining spins to be minus. Then this is a frozen configuration, and the set of plus
sites contains no 3-cube as a subset and similarly,
the set of minus sites contains no 3-cube as a subset.
\begin{figure}[h!]
    \centering
    \includegraphics[width=0.8\linewidth, trim=0cm 17cm 0cm 0cm]{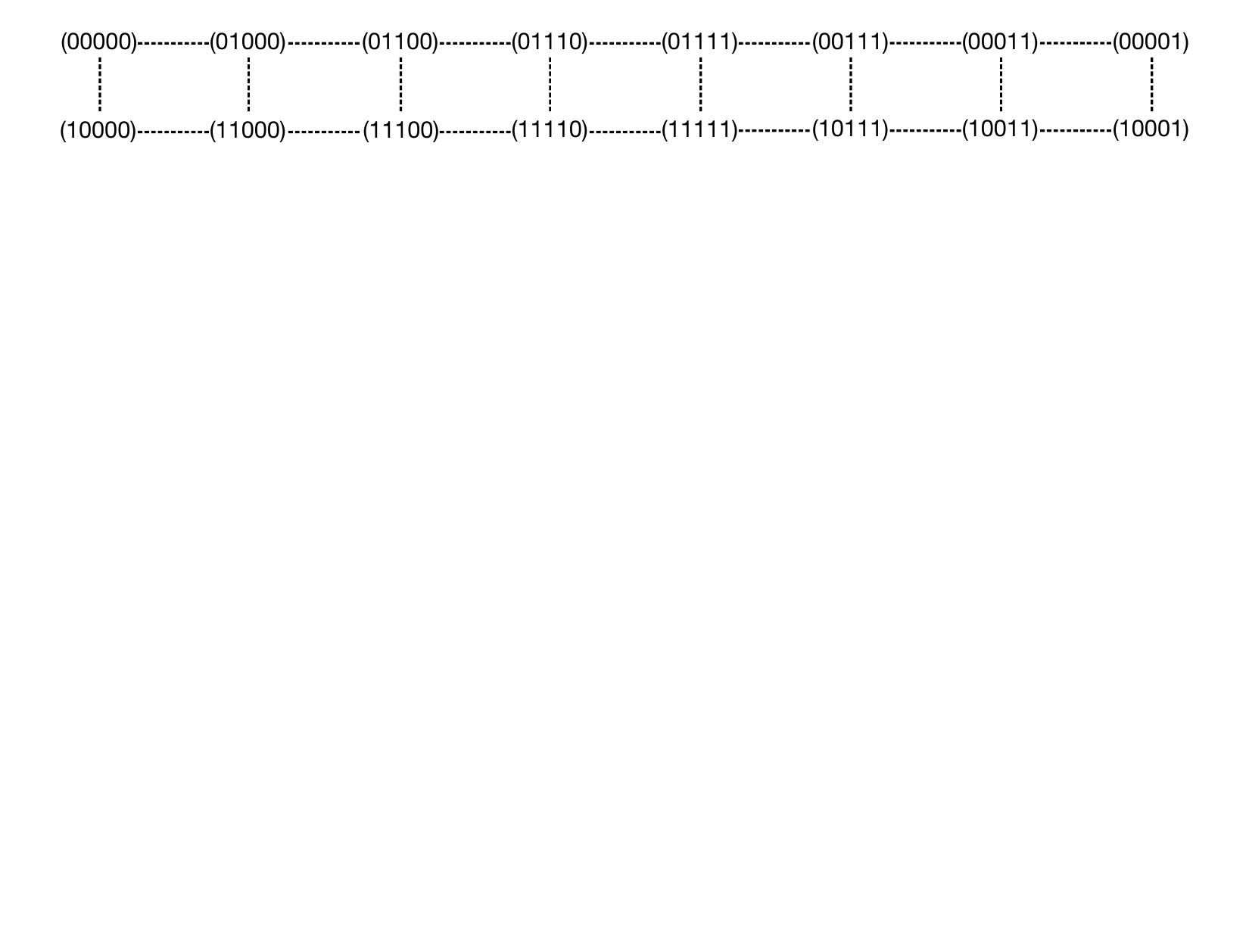}
    \caption{Coordinates of the plus spins in $Q_5$ that contains no 3-cubes}
    \label{fig:Figure16}
\end{figure}
\end{example}
\noindent The graph of the plus spins can be viewed as a ring of eight 2-site sets. In fact, its complement, the graph of the minus spins, shares exactly the same structure (See \cref{fig:Figure17} for a visualization). Both substructures are frozen because each spin has three neighbors it agrees with. Hence we have a frozen configuration in which the set of plus (minus) spins is a union only of $Q_{\lfloor \frac{d}{2}\rfloor}$'s.\\\newline
We can also do a similar $2$-renormalization on $Q_6$, but instead, we will have $Q_{2,j}$ on each vertex $j$ of $\Tilde{Q}_4$. This is a frozen state containing no $4$-dimensional hypercubes. This type of construction begins to fail on $Q_7$, which contains $4$-cubes.\\
\begin{figure}[h]
    \centering
    \includegraphics[width=0.5\linewidth, trim=0cm 17cm 0cm 0cm]{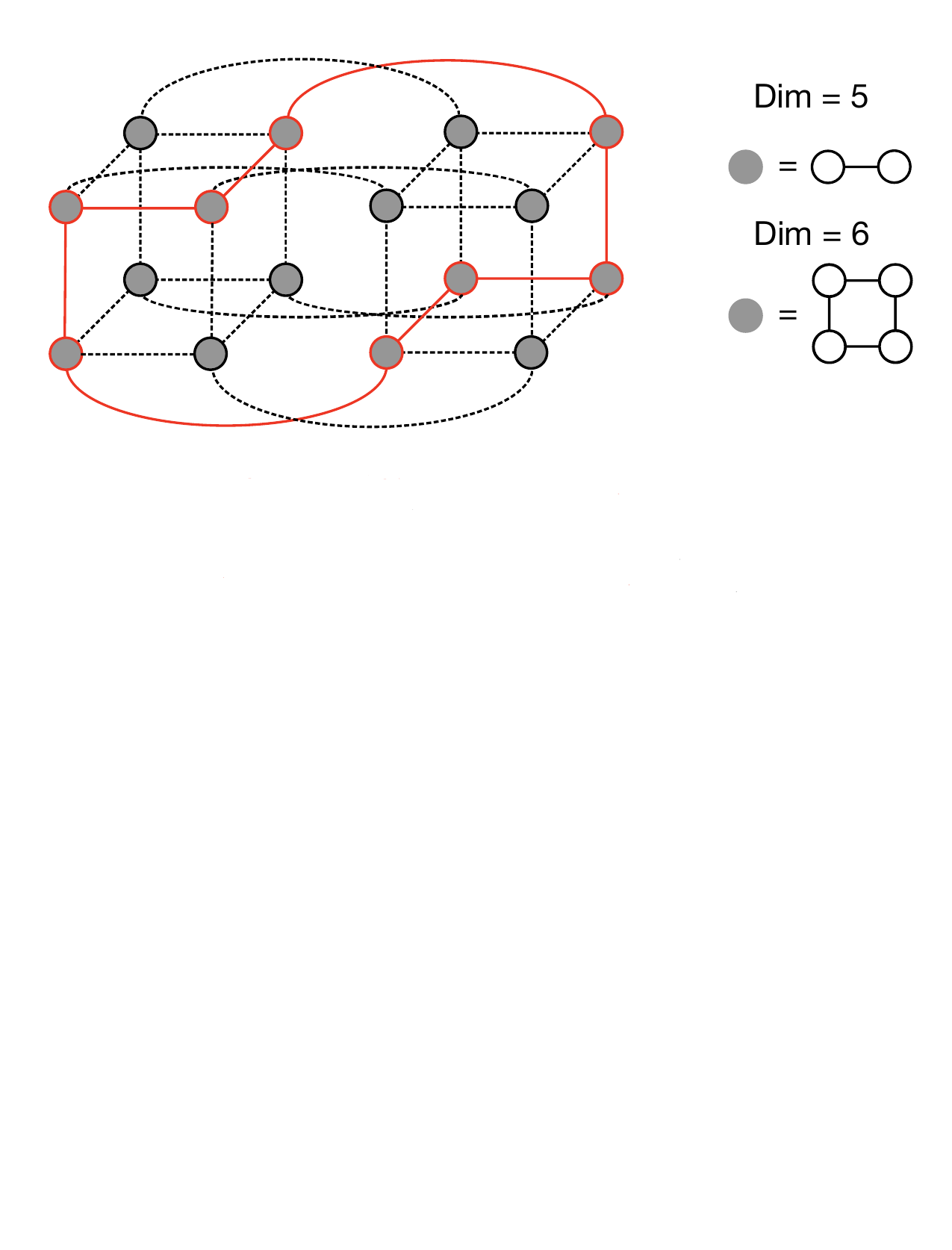}
    \caption{Visualization for $Q_5$ in \cref{fig:Figure16} $1$-renormalized into $\Tilde{Q}_4$, with a $1$-dim sub-cube $Q_{1,j}$ on each vertex $j$. A red-filled circle stands for a $1$-dim plus sub-cube (i.e.: two connected points), and a black-filled circle stands for a $1$-dim minus sub-cube. The red edges are the edges between plus spins, and the dashed edges are the other edges. The red loop is the structure described in \cref{fig:Figure16}. Similarly, one can form a black loop only using the minus spins.} 
    \label{fig:Figure17}
\end{figure}\\
\noindent Suppose now we want to construct a frozen configuration in which the set of plus (minus) spins is a union only of $Q_{\lfloor \frac{d}{2}\rfloor-a}$'s. Because the dimension in each sub-cube is lower, we  require a higher-dimensional structure connecting them. We first consider the $Q_{\lfloor \frac{d}{2}\rfloor-1}$ case, in which we choose a two-dimensional torus as an analog of the ring structure in \cref{exp:2}:
\setcounter{example}{2} 
\begin{example}\label{exp:3}

\begin{figure}[h!]
    \centering
    \includegraphics[width=0.9\linewidth, trim=0cm 10cm 0cm 2cm]{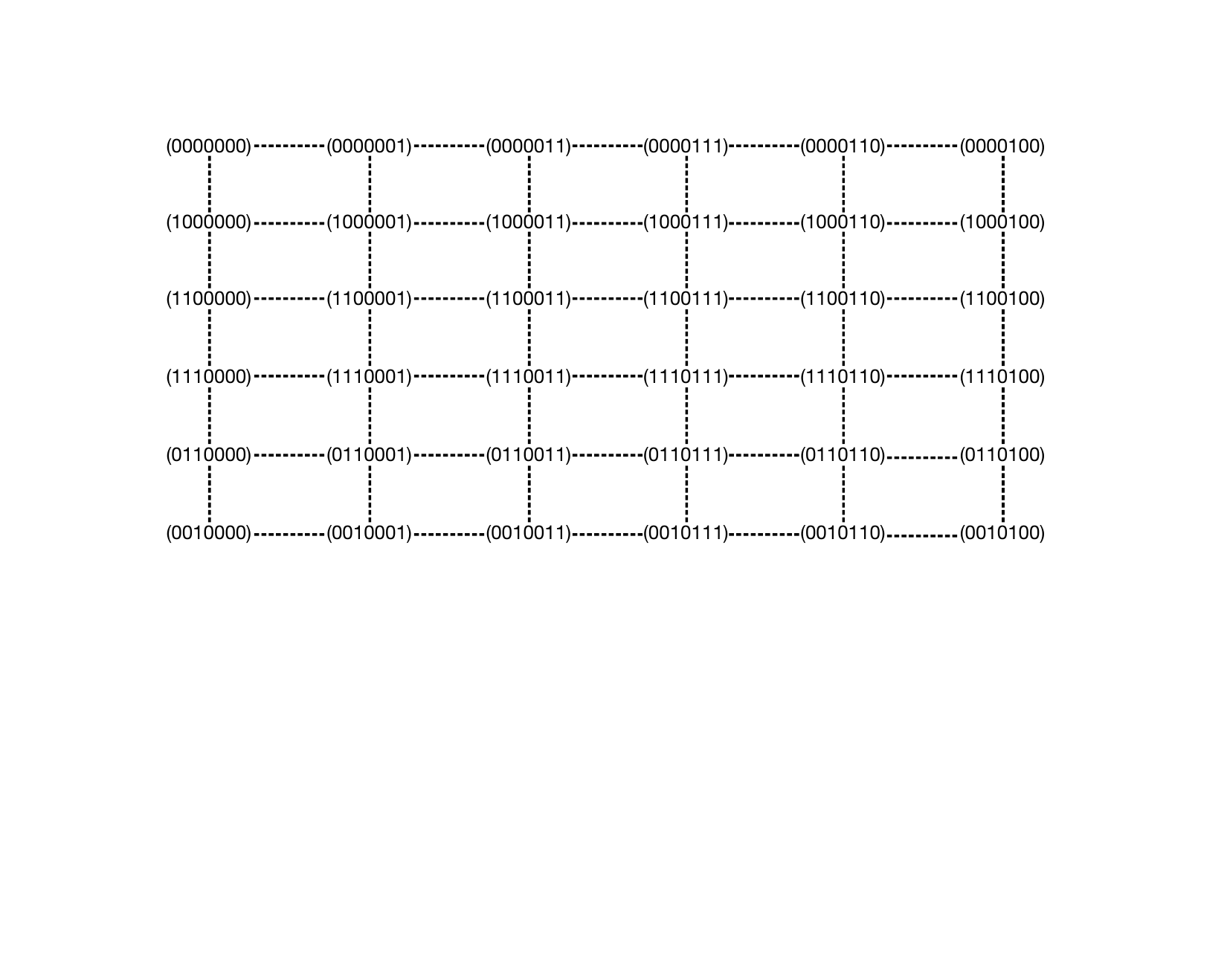}
    \caption{Coordinates of the plus spins in $Q_7$ that contains no 4-cubes}
    \label{fig:Figure18}
\end{figure}
On a $7$-dim hypercube $Q_7$, set all the spins on the sites in \cref{fig:Figure18} to be plus, and set all the remaining spins to be minus. Then this is a frozen configuration, and the set of plus
sites contains no $3$-cube as a subset.\\
Remark: In contrast to \cref{exp:2}, the set of minus spins here contains $4$-cubes.
\end{example}
\noindent Similarly as before, we can do a renormalization on $Q_8$, and assign each site in \cref{fig:Figure18} to be a $d=1$ sub-cube $Q_{1,j}$, which is frozen with no $4$-cubes contained.
So in principle, in order to construct a frozen state in which the set of plus (minus) spins is a union only of $Q_{\lfloor \frac{d}{2}\rfloor-a}$'s, we can choose the structure connecting them to be an $(a+1)$-dimensional torus embedded in $Q_d$. Finally, to ensure a successful embedding, let $l_i$ be the length of the torus in the $i$-th dimension.  We require~\cite{Havel1972}
\begin{equation}
    \sum_i^{a+1}\lfloor\log_2(l_i)\rfloor\leq d\, .
\end{equation}

\subsection{Blinker States}
\noindent We begin this section by describing an initial configuration consisting of a union of sub-cubes that must yield a blinker as a final state. For simplicity, we confine the discussion to eight dimensions, the lowest dimension where we observed blinker configurations; but a similar construction can be applied to any even $d\geq8$. We will assume the blinker resides at the origin~$\Vec{0}$, where it must have $4$ plus neighbors and $4$ minus neighbors. We assume spins at the nearest-neighbor sites $(1000|0000)$, $(0100|0000)$, $(0010|0000)$, $(0001|0000)$ to be plus, and spins at the sites  $(0000|1000)$, $(0000|0100)$, $(0000|0010)$, $(0000|0001)$ to be minus (the vertical line in these coordinates is for clarity to separate the first and last $d/2$ coordinates).\\\newline
The construction then proceeds as follows: at $t=0$, for each plus neighbor choose a $5$-dimensional sub-cube that includes the plus neighbor and make all its spins plus; for each minus neighbor, choose a $5$-dimensional sub-cube including the minus neighbor and make all its spins minus. These sub-cubes need to be chosen such that the sets of plus and minus spins are disjoint. The spin values in the complement of the union of these eight $5$-dimensional sub-cubes can be arbitrary.  Now let the above initial configuration evolve through using the usual Glauber dynamics. Since a fully magnetized sub-cube of dimension greater than $d/2$ is frozen, the union of plus sub-cubes and the union of minus sub-cubes each remains frozen for all times, so as $t\rightarrow\infty$, the spin at $\Vec{0}$ becomes a 1-blinker. \\\newline
Next we give an explicit construction of  disjoint sets of plus and minus sub-cubes each of which contain a neighbor of the origin. We use the notation * to indicate that the given coordinate can be either $0$ or $1$, so e.g.~a (*1**$|$**00) is a $5$-dimensional sub-cube embedded in $Q_8$. Here is one possible configuration of four plus $5$-dimensional sub-cubes and four minus $5$-dimensional sub-cubes that have no overlap:
\[
\text{plus $5$-dim sub-cubes:  (1***$|$00**), (*1**$|$**00), (**1*$|$00**), (***1$|$**00)}; 
\]
\[
\text{minus $5$-dim sub-cubes:  (*0*0$|$1***), (*0*0$|$*1**), (0*0*$|$**1*), (0*0*$|$***1)}.
\]
A fully explicit construction of a blinker is obtained setting the complement of the union of the four plus $5$-dimensional sub-cubes above to be have all minus spins.
This is because the complement, except for the origin, is also a union of minus $5$-dimensional sub-cubes, and thus is also frozen. A more detailed proof is given in the \cref{app:constructblinker}.\\\newline
Clearly, the existence of blinkers depends on their local environment, especially their nearest and next-nearest neighbors, and so is strongly dependent on dimension. It is natural to ask what is the lowest dimension in which a single-site blinker can exist?
\setcounter{theorem}{2} 
\begin{theorem}\label{thm:3}
There are no blinker configurations on a hypercube $Q_d$ when $d<8$.
\end{theorem}

\begin{proof}
Using the same setup as before, we write a site of the hypercube as a string of 0's and 1's, and assume the spin at $\Vec{0}$ is a blinker. Each nearest neighbor of the origin has a single 1 in its binary string.  Without loss of generality, we divide the string from the middle, with $d/2$ places on the left and $d/2$ places on the right. We assign plus spins to all sites with a 1 on the left (call them the L sites), and assign minus spins to all $d/2$ nearest neighbor sites with a 1 on the right (the R sites).\\\newline
\noindent There are three types of next-nearest neighbors of the origin: The first has any permutation of two 1's on the left with all the rest zeroes, which we shall call Type 1 sites; the second has any permutation of two 1's on the right with all the rest zeroes, called Type 2 sites.  Type 3 has any permutation of a single 1 on the left and a single 1 on the right with all the rest zeroes. Each Type 1 site has two neighboring L sites, thus the L sites are most stable if we assign their spins to be plus. Similarly, we assign the Type 2 sites to be minus.   \\\newline
\noindent So far each L site has $d/2 - 1$  plus neighbors and each R site has $d/2 - 1$ minus neighbors, so to be stable every L site needs at least two additional next-nearest neighbors with a plus spin, and similarly for the R sites. Since all of the Type 1 \& 2 sites have already been used, then given that each nearest neighbor requires its own distinct set of two additional neighbors having spins of the same sign, we require a minimum of $2d$ Type 3 sites. Since we know that there are $\frac{d^2}{4}$ Type 3 sites in total, we obtain the inequality $\frac{d^2}{4}\geq2d$ to ensure that there are enough next-nearest neighbors to stabilize the nearest neighbors. We then conclude that $d\geq 8$.
\end{proof}
\begin{figure}[h!]
    \centering
    \includegraphics[scale=0.40]{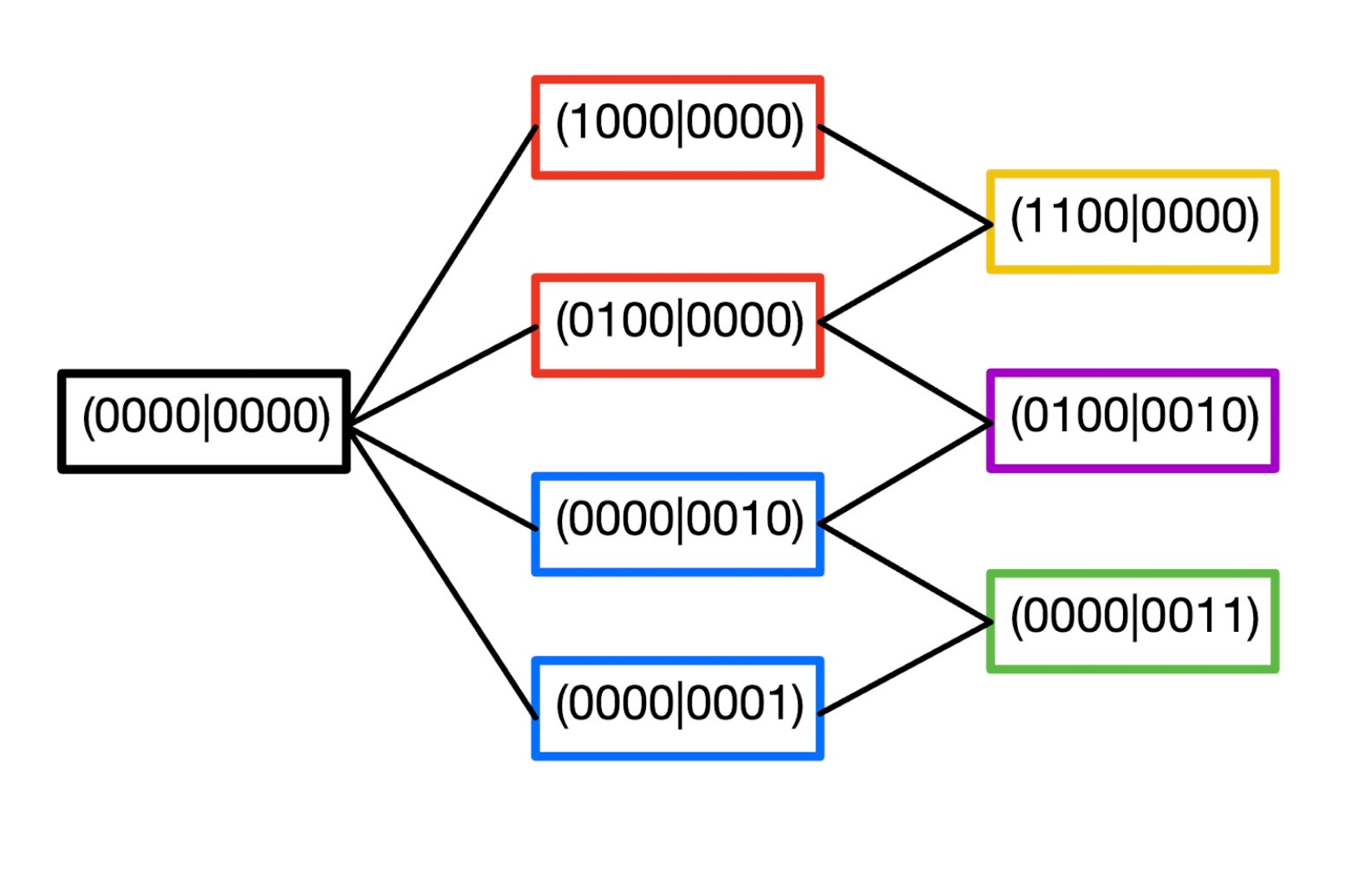}
    \caption{Counting argument for the existence of a single-site blinker. The black rectangle is the blinker, the red rectangles are L sites, the blue rectangles are  R sites, the yellow rectangle is a Type 1 site, the green rectangle is a Type 2 site, and the purple rectangle is a Type 3 site.}
    \label{fig:Figure21}
\end{figure}

\noindent In our simulations, we've observed connected blinkers with several structures. To appreciate these structure, we present a few examples. The simplest case is the 2-blinker.  Figure \ref{fig:Figure19} shows the local configuration of a 2-blinker and the possible transitions between its three final states. Figure \ref{fig:Figure20} shows the local configuration of a 4-site tree blinker and the dynamics of its five final states. 
\begin{figure}[h!]
    \centering
    \includegraphics[scale=0.6, trim=0cm 17cm 0cm 0cm]{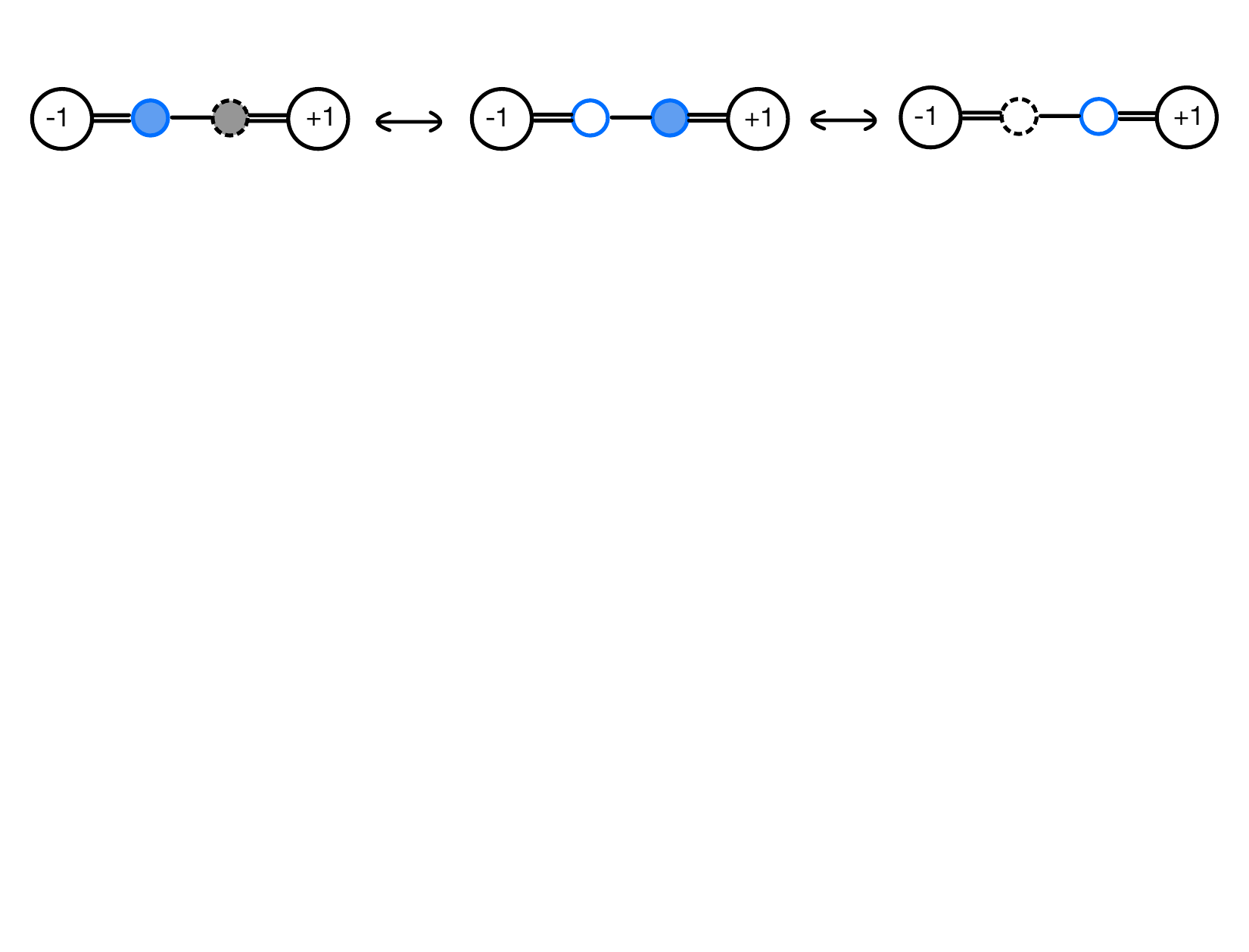}
    \caption{Illustration for the local configuration of a double-site blinker and its dynamics. Filled circles are plus spins and empty circles are minus spins. Blue circles are spins that are currently free to flip. Dashed circles are spins that are currently frozen, but can potentially become free to flip in the future. A circle containing a number describes the remaining neighbors of the attached blinker site.  The number is the net magnetization of those neighbors. Each configuration is dynamically accessible from the configuration(s) next to it.}
    \label{fig:Figure19}
\end{figure}
\begin{figure}[h!]
    \centering
    \includegraphics[scale=0.39]{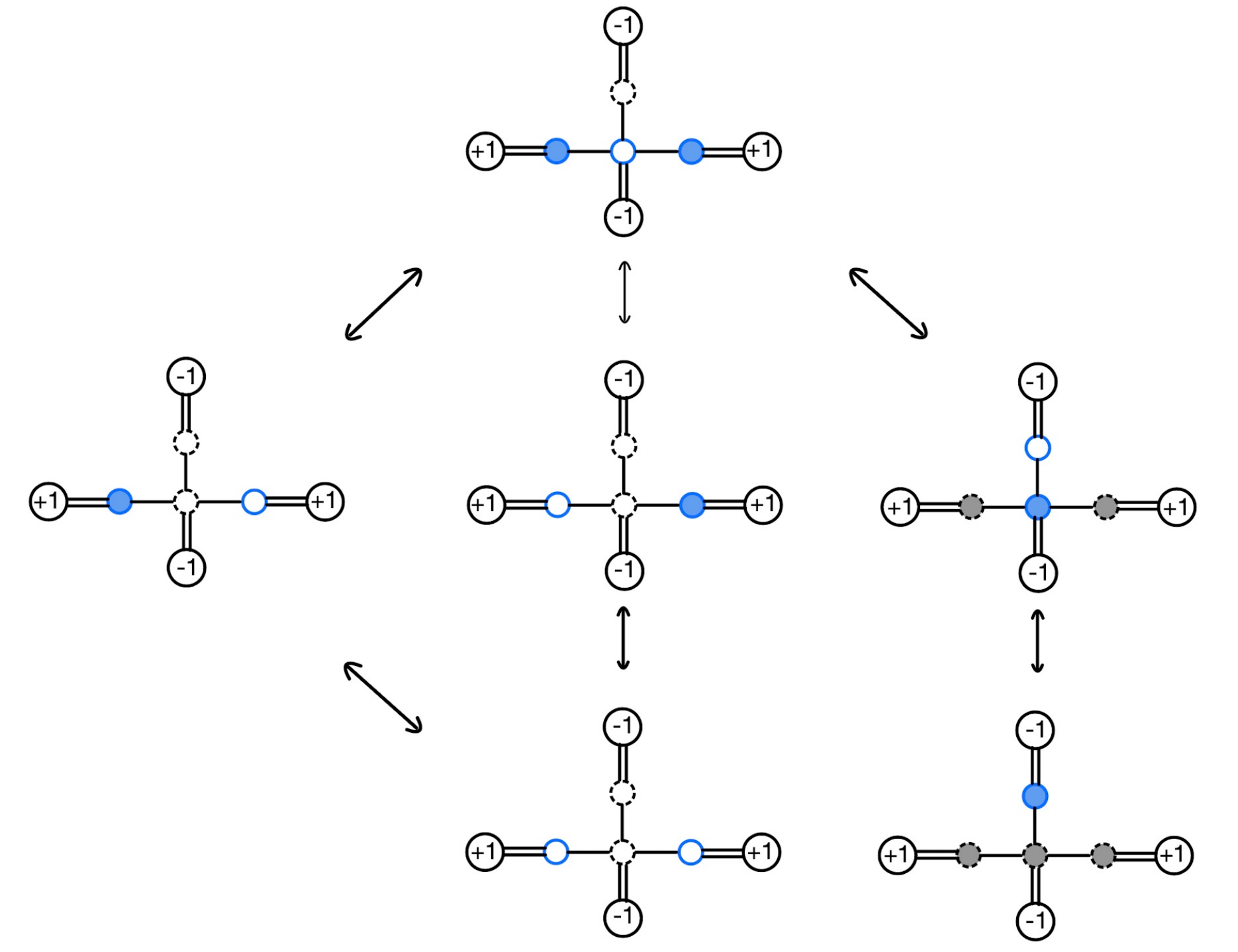}
    \caption{Illustration for the local configurations of a tree blinker and its dynamics (see \cref{fig:Figure19} for additional explanation). Each configuration is dynamically accessible by the arrow from the configuration(s) next to it.}
    \label{fig:Figure20}
\end{figure}
\begin{figure}[h!]
    \centering
    \includegraphics[scale=0.65]{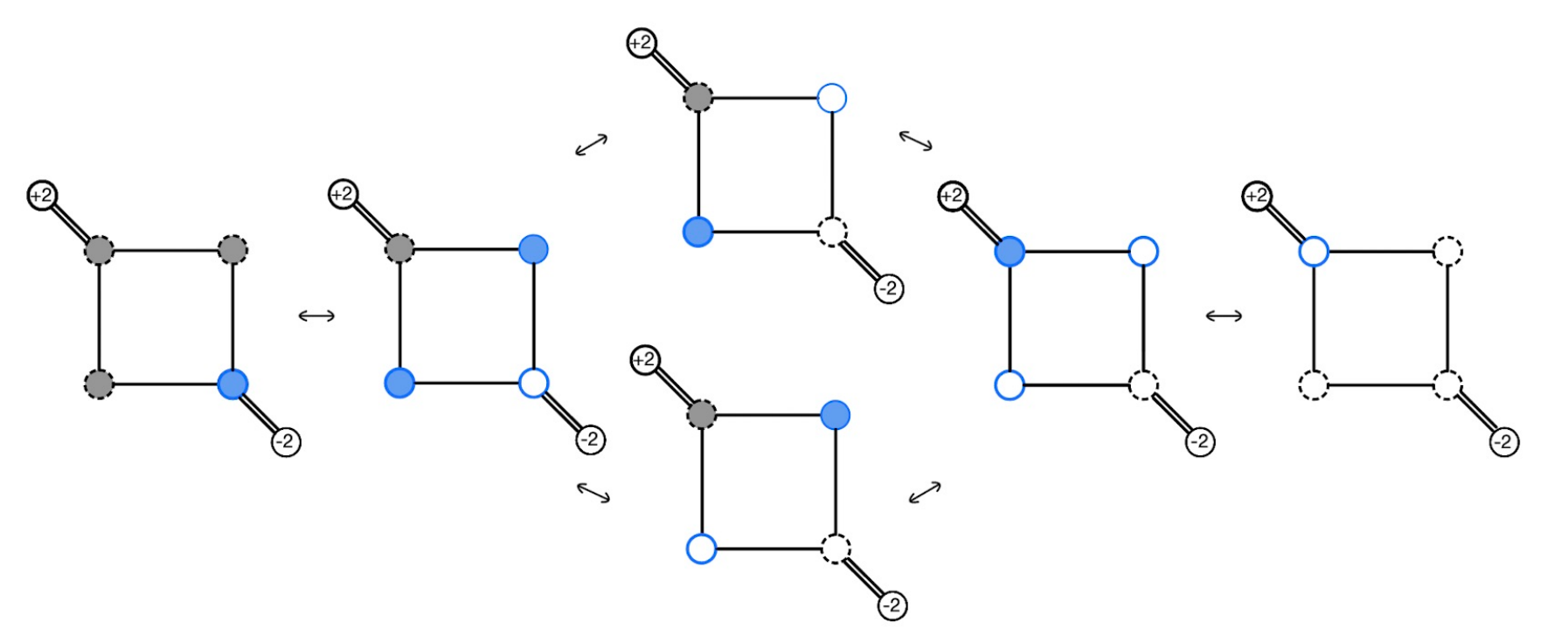}
    \caption{Illustration for the local configuration of a loop blinker and its dynamics (see \cref{fig:Figure19} for additional explanation). The net magnetizations of the neighbors  of both the upper right and lower left blinker spins is zero. Each state commutes with the state next to it.}
    \label{fig:Figure24}
\end{figure}\\
\noindent Blinker configurations can also appear as closed circuits in sufficiently high dimension. \cref{fig:Figure24} shows the structure for a simple closed-loop blinker. The top-left-corner site has $2$ more plus-neighbor spins than minus-neighbor spins, and the bottom-right-corner site has $2$ more minus-neighbor spins. The sites at the top right and bottom left have an equal number of plus-neighbor and minus-neighbor spins. One possible dynamic (see \cref{fig:Figure24}) for the loop blinker is the following. Assume all four spins are initially $+1$. The spin on the bottom left can flip with zero energy; assume it does so. Now the spins on the top right and bottom left can each perform a zero-energy flip. If both then flip to~$-1$, the top left spin can also do a zero-energy flip. Now all four spins have changed to $-1$, and the procedure can be reversed to arrive back at an all-plus configuration. 
\newpage
\noindent In \cref{exp:4} below, we present a construction of a blinker configuration containing a 2-blinker on a $10$-dimensional hypercube using a $2$-renormalization. In \cref{app:constructblinker}, we also provide an explicit construction for this configuration using disjoint sets of plus and minus sub-cubes. One can repeat this construction to obtain a blinker on a non-branching acyclic path of length $n$, and a similar approach can be used to construct a configuration containing a tree blinker.
\begin{example}\label{exp:4}
A final state with a 2-blinker $\sigma^{10,2b}$ in  the $10$-dimensional hypercube $Q_{10}$ -- see \cref{fig:Figure23}. In the diagram,  $\sigma^{8,b}$ refers to the  final state of $Q_8$  containing a single 1-blinker.  Note that this final state consists of two final configurations since the blinker site can take two spin values. The 2-blinker final state contains only three configuration.  The configuration where the upper blinker site is minus and the lower blinker site is plus is unstable but after a single flip it enters the final state.
\\
\begin{minipage}{0.37\textwidth}
\centering
\includegraphics[scale=0.6, trim=4cm 19cm 6cm 0cm]{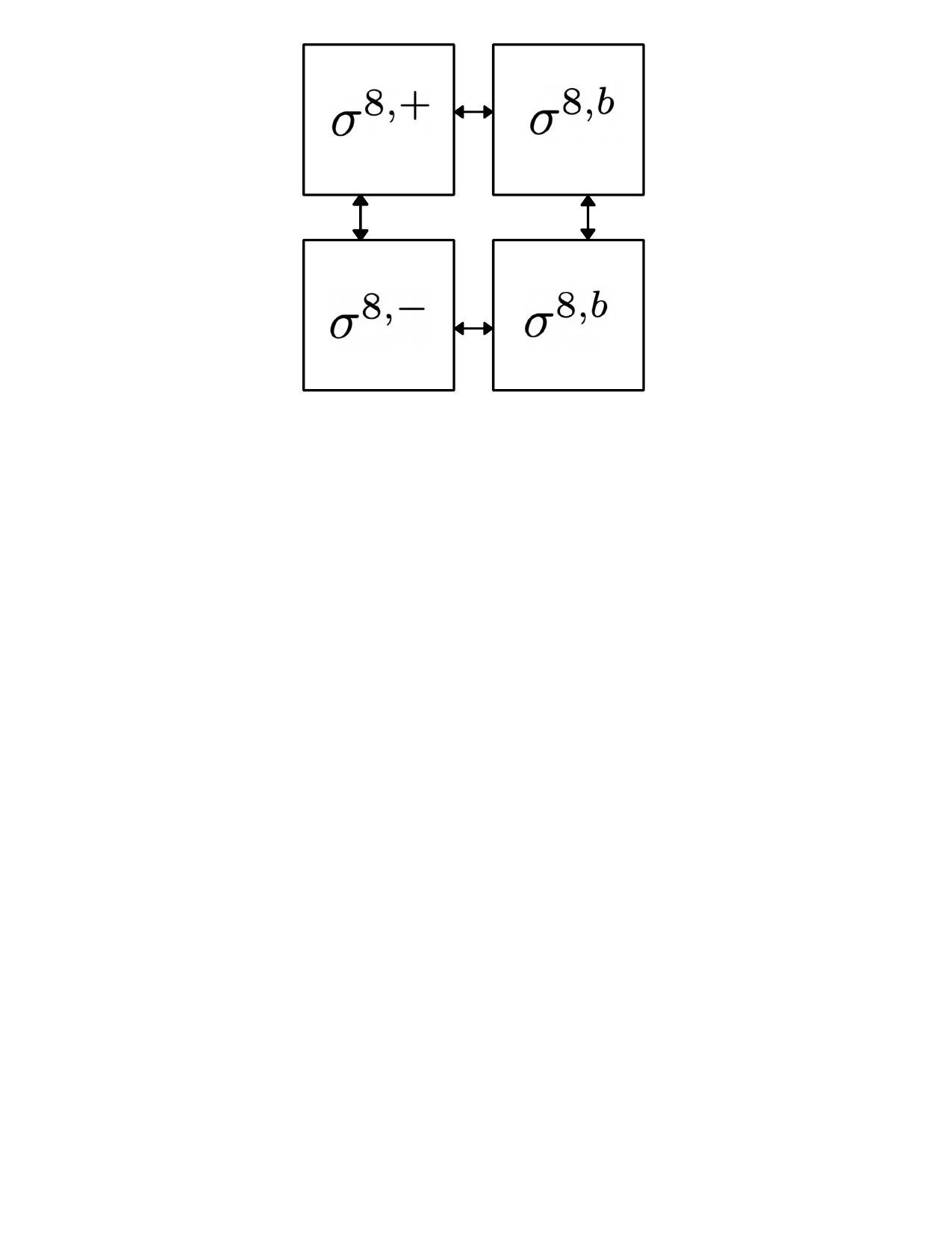}
\end{minipage}%
\hfill
\begin{minipage}{0.50\textwidth}
\begin{equation}
\sigma^{10,2b}_{(i, i_{9}, i_{10})} = 
\begin{cases} 
      \sigma^{8,+}_{i} & \text{if } (i_{9}, i_{10}) = (0, 0), \\[8pt]
      \sigma^{8,b}_{i} & \text{if } (i_{9}, i_{10}) = (0, 1), \\[8pt]
      \sigma^{8,-}_{i} & \text{if } (i_{9}, i_{10}) = (1, 0), \\[8pt]
      \sigma^{8,b}_{i} & \text{if } (i_{9}, i_{10}) = (1, 1),
\end{cases}
\end{equation}
where $i\in Q_8$, and $(i_9, i_{})\in\{0,1\}^2$
\end{minipage}
\captionof{figure}{Graphs for construction of a 2-blinker in \cref{exp:4}. Each square represents an $8$-dim hypercube, a double-headed arrow between two hypercubes means connections between all spins with the same coordinates. In particular, the 1-blinker in the upper right sub-cube is connected to the 1-blinker in the bottom right sub-cube. The plus sub-cube and the minus sub-cube are clearly stable. Except for the 1-blinker, each of the spins in the upper right sub-cube gains an extra spin aligning with it by attaching to the bottom right sub-cube, hence will remain frozen when attaching to the plus sub-cube. Similarly, each of the spins in the bottom right sub-cube will remain frozen. The 1-blinker in the upper right sub-cube gains an extra plus spin, and the 1-blinker in the bottom right sub-cube gains an extra minus spin, together forming a 2-blinker.}
\label{fig:Figure23}
\end{example}

\noindent We now present a construction (see \cref{fig:Figure22}) for a blinker configuration containing such a loop blinker on $Q_{12}$. We first pick a frozen configuration $\sigma^{10,a}$ with a specific site $i$ having $6$ plus neighbor spins and $4$ minus neighbor spins, which we give its explicit construction in \cref{app:constructblinker}. We can construct a blinker configuration $\sigma^{10,b}$ with a single blinker at site $j$ on $Q_{10}$ from a blinker configuration and an arbitrary frozen configuration on $Q_8$ using \cref{thm:1}. Using the permutation invariance of the hypercube, we can without loss of generality assume $i=j$. Then by construction (10) in~\cref{thm:1}, we can combine these two configurations into a single configuration $\sigma^{12}$ on $Q_{12}$. Now the single blinkers, the spin on site $i$, and their flipped spins form a loop blinker.\\\newline
\begin{figure}[h!]
    \centering
        \includegraphics[scale = 0.6,trim=10cm 18cm 0cm 1cm, clip]{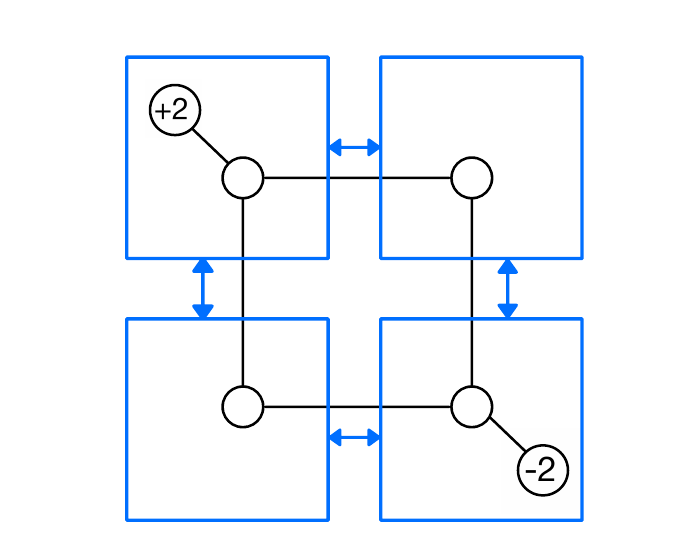}
    \caption{Construction for a blinker configuration that contains a loop blinker in $d=12$ using four $Q_{10}$'s. The upper right $Q_{10}$ is a blinker configuration $\sigma^{10,b}$ that contains a 1-blinker on site $i$. The bottom left $Q_{10}$ is the spin-flipped version of $\sigma^{10,b}$. The upper left $Q_{10}$ is a frozen configuration $\sigma^{10,a}$, with site $i$ having $6$ plus neighbors and $4$ minus neighbors. The bottom right is the spin-flipped version of $\sigma^{10,a}$.}
    \label{fig:Figure22}
\end{figure}

\section{Discussion}
\noindent We have seen that the evolution and final states of zero-temperature Ising dynamics are surprisingly rich in high dimensions. Our numerical and analytical results raise many questions, severla of which we state below as conjectures and support them with heuristic arguments. Our numerical results up to dimension 17 suggest that the final state for large $d$ is almost never a ground state.
\begin{conjecture}\label{cjt:1}
The probability for a $d$-dimensional hypercube to end up in a ground state goes to zero as $d\rightarrow\infty$:
\begin{equation}
\lim_{d\rightarrow\infty}\mathbb{P}_{\sigma(0), \omega}(\{G^d_\infty = \sigma^{d,+}\text{ or }\sigma^{d,-}\})=0.   
\end{equation}
\end{conjecture}

\noindent Numerical results (\cref{fig:Figure2}) strongly support this, suggesting in fact that $\mathbb{P}_{\sigma(0), \omega}(\{G^d_\infty = \sigma^{d,+}\text{ or }\sigma^{d,-}\})$ decays to zero exponentially fast. \noindent\cref{thm:2} states that the number of frozen states grows at least exponentially fast with a power factor of $2^{\frac{d}{2}}$. As $d$ grows, the rapidly increasing number of frozen states will trap the dynamics, preventing the hypercube from entering a ground state. We therefore expect the probability for a $d$-dimensional hypercube to enter a ground state to vanish as $d\rightarrow\infty$.\\\newline
Now suppose that for some large $d$ the probability to dynamically find a ground state is $p$. We know $p$ is strictly between $0$ and $1$, since there's positive probability for the hypercube to start with either a frozen state or a ground state. If we run the dynamics on two uncoupled $d$-dimensional hypercubes, the probability for the system to enter a ground state is $p^2$. Now consider a $(d+1)$-dimensional hypercube composed of two $d$-dimensional sub-cubes, coupled in the sense that  corresponding sites on each are connected. If we run the dynamics on the $(d+1)$-dimensional hypercube, we expect the results to resemble that from the uncoupled dynamics, because the density of perturbable spins vanishes as $d\to\infty$. By ``perturbable spin'' we mean one for which a change in the state of a single neighboring spin  can change whether the original flips (or changes the direction of the flip). Consequently, we expect the probability for the $(d+1)$-dimensional hypercube to enter a ground state to be close to $p^2$ for large $d$.\\\newline
The above argument comparing the coupled and uncoupled dynamics depends on the strength of the correlation between its two $d$-dimensional sub-cubes. We address this in the next conjecture. In our simulation, as the dimension increases, the final magnetization concentrates around zero, suggesting a central-limit-type behavior. 
\begin{conjecture}\label{cjt:2}
(Central-Limit-Type Behavior) The final magnetization $M_\infty(N) = M_\infty(d)$ converges in distribution to a standard Gaussian under some rescaling, i.e., there exists a function $f(N)$ such that:
\begin{equation}
\frac{1}{f(N)}M_\infty(N)\xrightarrow[N\rightarrow\infty]{(d)}\mathcal{N}(0,1)\, .
\end{equation}
\end{conjecture}
\noindent Note that if \cref{cjt:2} is true, then \cref{cjt:1} follows. Our numerical study (\cref{fig:Figure3} (a) and (b)) suggests that the correct scale $f$ should satisfy $\sqrt{N} \ll f(N) \ll N$. For simplicity, we consider $d$ to be odd to avoid the presence of blinkers. Then given two spins $\sigma^d_{i_1}$ and $\sigma^d_{i_2}$ on a $d$-dimensional hypercube $Q_d$ with odd $d$, we have by symmetry that their 2-spin correlation function depends only on their separation $l$, where $l = \mathbf{d}_h(i_1,i_2)\leq d$ and $\mathbf{d}_h$ is the Hamming distance. We then define the equal-time 2-spin correlation function as $t\rightarrow\infty$ as
\begin{equation}
    C^d(l) = \lim_{t\rightarrow\infty}{\rm Cov}(\sigma^d_{i_1}(t),\sigma^d_{i_2}(t)),\text{ where }l = \mathbf{d}_h(i_1,i_2)\, .
\end{equation}
\noindent Next consider a $k$-renormalization of $Q_d$ into $Q_{d-k}$, where each vertex $j$ corresponds to a $k$-dimensional sub-cube $Q_{k,j}$ with $k\geq \lfloor\frac{d}{2}\rfloor+1$. For $j_1,j_2\in Q_{d-k}$, we consider the final magnetization $S_{k,j_1}$ and $S_{k,j_2}$ for sub-cubes $Q_{k,j_1}$ and $Q_{k,j_2}$:
\begin{equation}
    S_{k,j_1} = \lim_{t\rightarrow\infty}\sum_{i\in Q_{k,j_1}}\sigma^d_i(t)\, .
\end{equation}
Therefore the correlation between the final magnetization $S_{k,j_1}$ and $S_{k,j_2}$ is
\begin{equation}
    \lim_{t\rightarrow\infty}{\rm Corr}(S_{k,j_1},S_{k,j_2}) = \lim_{t\rightarrow\infty} \frac{1}{\sqrt{{\rm Var}(S_{k,j_1}){\rm Var}(S_{k,j_2})}}\sum_{i_1\in Q_{k,j_1}}\sum_{i_2\in Q_{k,j_2}}{\rm Cov}(\sigma^d_{i_1}(t),\sigma^d_{i_2}(t))
\end{equation}
where
\begin{equation}
    {\rm Var}(S_{k,j_1}) = \sum_{i_1,i_2\in Q_{k,j_1}}{\rm Var}(\sigma^d_{i_1}(t),\sigma^d_{i_2}(t))\, .
\end{equation}
\noindent A standard argument(\cite{Newman1980}) for proving  CLT behavior involves showing that the correlation between the final magnetizations $S_{k,j_1}$ and $S_{k,j_2}$ vanishes as $d\rightarrow\infty$, making use of the Fortuin–Kasteleyn–Ginibre (FKG) inequality\cite{Harris1977}:
\begin{equation}
{\rm Cov}(f(\sigma^d(t)), g(\sigma^d(t))\geq 0, \text{ where $f,g$ are increasing functions}\, .
\end{equation}
But here there is a problem preventing use of this approach: the FKG inequality usually requires configurations with positively correlated initial conditions, which is not the case for initial configurations with zero magnetization. An alternate initial configuration which satisfies the conditions for use of the FKG inequality is the symmetric i.i.d initial condition. To prove the asymptotic vanishing of the correlation between the final magnetization $S_{k,j_1}$ and $S_{k,j_2}$ it would be useful to have an explicit formula for $C^d(l)$. We have obtained explicit formulas for $C^3(l)$ and $C^4(l)$, but the computation quickly becomes unwieldy as $d$ grows larger due to the growing number and complexity of frozen states.\\\newline
Our numerical results (see \cref{fig:Figure11}) show hypercubes enter blinker states with high probability on high even dimensions ($d=18$); our next conjecture addresses this.
\begin{conjecture}
The probability for an even $d$-dimensional hypercube to enter a blinker state goes to one as $d\rightarrow\infty$, i.e.:
\begin{equation}
\lim_{d \text{ even},d\rightarrow\infty}\mathbb{P}_{\sigma(0), \omega}(|G^d|>1)=1\, .
\end{equation}
\end{conjecture}
\noindent Intuitively, this is because the number of sites grows exponentially as dimension increases, providing more room for blinker configurations. The results shown in~\cref{fig:Figure4} adds additional support since they show that on high-dimensional hypercubes, the distribution of local fields $m_\infty(d)$ is supported on all allowed values, rather than concentrating somewhere away from zero. 

\begin{conjecture} 
In double-copy dynamics, the final overlap of the twins is strictly between $0$ and $1$ as $d\rightarrow\infty$, i.e.:
\begin{equation}
\lim_{d\rightarrow\infty}q_\infty(d)\in(0,1)\, . 
\end{equation}   
\end{conjecture}
\noindent In both even and odd dimensions, numerical results (\cref{fig:Figure13}) suggest that $q_\infty(d)$ increases monotonically, with a possible limit around $\frac{1}{2}$. Given that the number of frozen states grows exponentially fast with the dimension, the hypercube will have many possible final states whose overlap is small; hence it seems implausible for the final overlap to approach~one as $d\rightarrow\infty.$ At the same time, the low average number of flips per site implies that it would be difficult for double-spin dynamics to evolve toward two final configurations that are substantially different; therefore we believe that the final overlap will be strictly positive as $d\rightarrow\infty$.

\appendix
\label{app:finalstate}
\section{Final State Algorithm for obtaining $G_\infty$}
\label{app:finalstate}

\begin{enumerate}
    \item Run the dyanmics until the active list $a(t)$ is small and contains only spins that can perform a zero-energy flip.  (This step is not necessary but will reduce the number of restarts.)
    \item Perform a Breadth First Search (BFS) of the graph, $\mathcal{G}(t)$ whose vertices are elements of $G(t)$ and whose edges are allowed, single spin flip transitions.   
    \begin{enumerate}
        \item If an energy lowering spin flip is encountered in the BFS then $G(t) \neq G_\infty$ and the dynamics must be run to a later time before the final state algorithm is run again.
        \item If the BFS terminates without encountering an energy lowering spin flip then the final state has been identified, $G(t) = G_\infty$ and $\mathcal{G}(t)= \mathcal{G}_\infty$. 
    \end{enumerate}
\end{enumerate}

\section{Explicit Construction for Configurations}
\label{app:constructblinker} 
\subsection{A blinker configuration at d = 8}
\noindent Set the following four $5$-dimensional sub-cubes to be all plus, 
\[
\text{(1***$|$00**), (*1**$|$**00), (**1*$|$00**), (***1$|$**00)}.
\]
Let $M$ be the union of these plus sub-cubes together with the origin, which can be either plus or minus.
Set the complement of $M$ to be all minus and call it $N$.\\
\newline
\textit{Claim:} The above configuration of an $8$-dimensional cube comprises a final state with a 1-blinker site at the origin $\Vec{0}$.
\\
\newline
First note that the origin has 4 plus and 4 minus neighbors as required for a 1-blinker.
We now give an explicit construction for $N$ and show that it is frozen.  We assert that $N$ is the union of the following four $5$-dimensional subcubes and four $6$-dimensional subcubes:
\[
\text{(*0*0$|$1***), (*0*0$|$*1**), (0*0*$|$**1*), (0*0*$|$***1)}.
\]
\[
\text{(****$|$1*1*), (****$|$1**1), (****$|$*11*), (****$|$*1*1)}.
\]
It is straightforward to check that each of these subcubes is disjoint from each plus subcube in $M$. Therefore, the union of these minus subcubes is contained in $N$. By explicit counting, $|M| = 88$ and therefore $|N| = 256-88=168$.  By a straightforward but tedious count, one can see that the cardinality of the proposed construction indeed has cardinality $168$. Finally, the proposed construction of $N$ is frozen  since it is composed entirely of subcubes with dimensions greater or equal to 5.
\subsection{A frozen configuration at d = 10 with a specific site i having 6 plus neighbors and 4 minus neighbors}
\noindent By symmetry, without loss of generality, we can assume $i = \Vec{0}$. We use a similar string annotation as previous for the blinker configuration at $d = 8$, and again we use $*$ to indicate that the given coordinate can be either 0 or 1. But this time we divide the string such that there are 6 places on the left and 4 places on the right of the vertical line. We assign plus spins to all 6 sites with a 1 on the left, and assign minus spins to all 4 sites with a 1 on the right. At $t=0$, for each plus spin we assign a $6$-dim sub-cube with all plus spins, and similarly for minus spins: 
\[
\begin{aligned}
\text{plus } 6\text{-dim sub-cubes:} \quad & \text{Type I:  (1**;***$|$*000), (*1*;***$|$*000), (**1;***$|$*000)}; \\
& \text{Type II: (***;1**$|$000*), (***;*1*$|$000*), (***;**1$|$000*)};
\end{aligned}
\]

\[
\begin{aligned}
\text{minus } 6\text{-dim sub-cubes:} \quad & \text{Type I:  (***;000$|$***1), (***;000$|$**1*), (***;000$|$*1**)}; \\
& \text{Type II: (000;***$|$1***)};
\end{aligned}
\]
\noindent Sub-cubes are labeled Type I and II to help see that the plus and minus 6-dim sub-cubes are disjoint. We choose the complementary spins to be arbitrary and let this initial configuration evolve through the Glauber dynamics. Then we obtain the configuration as desired.
\subsection{A two-site blinker configuration at d = 10}
\noindent We assume the sites for the two blinkers to be (0000000000) and (1000000000). We further assume the 5 sites from (010000$|$0000) to (000001$|$0000) to be all plus, 4 sites from (000000$|$1000) to (000000$|$0001) to be all minus, and assume the 4 sites from (11000$|$00000) to (10001$|$00000) to be all plus, 5 sites from (10000$|$10000) to (10000$|$00001) to be all minus. At $t=0$, for the plus neighbors of (0000000000) we arrange six $6$-dim sub-cube with all plus spins, and four $5$-dim sub-cube with all minus spins for the minus neighbors. Similarly, for the plus neighbors of (1000000000) we assign four $5$-dim sub-cube with all plus spins, and a six $6$-dim sub-cube with all minus spins for its minus neighbors: 
\[
\begin{aligned}
\text{plus } 6\text{-dim sub-cubes for (0000000000):} \quad & \text{Type I:  (0,1**,**$|$**,00), (0,*1*,**$|$**,00), (0,**1,**$|$**,00)}; \\
& \text{Type II: (0,**,1**$|$00,**),(0,**,*1*$|$00,**), (0,**,**1$|$00,**)};
\end{aligned}
\]
\[
\begin{aligned}
\text{plus } 5\text{-dim sub-cubes for (1000000000):} \quad & \text{Type I:  (1,1*,**$|$**,000), (1,*1,**$|$**,000)}; \\
& \text{Type II: (1,**,1*$|$000,**), (1,**,*1$|$000,**)};
\end{aligned}
\]
\[
\begin{aligned}
\text{minus } 6\text{-dim sub-cubes for (1000000000):} \quad & \text{Type I: (1,**,00$|$**,**1), (1,**,00$|$**,*1*),(1,**,00$|$**,1**)}; \\
& \text{Type II: (1,00,**$|$**1,**), (1,00,**$|$*1*,**),(1,00,**$|$1**,**)};
\end{aligned}
\]
\[
\begin{aligned}
\text{minus } 5\text{-dim sub-cubes for (0000000000):} \quad & \text{Type I:  (0,**,000$|$**,*1), (0,**,000$|$**,1*)}; \\
& \text{Type II: (0,000,**$|$*1,**), (0,000,**$|$1*,**)};
\end{aligned}
\]
Note that the 5-cubes are not stable by themselves and are stabilized by 6-cubes of the same sign.  For example, the plus 5-cube (11***$|$**000) is stabilized by the plus 6-cubes in (01****$|$**00). This requires that every site in the set (11***$|$**000) has a neighbor in the set (01****$|$**00), which is to say that all the coordinates except the first agree. In the above list, for convention, the first two Type I(or II) plus(or minus) 6-dim sub-cubes stabilizes the Type I(or II) plus(or minus) 5-dim sub-cubes. Finally, we choose the complementary spins to be arbitrary and let this initial configuration evolve through the Glauber dynamics.
We then obtain the configuration as desired. In particular, if we set the complementary spins to be all minus (or plus), numerically we've verified that such configuration is already frozen.\\\newline
\noindent Remark: For even $d\geq 12$, the stabilization of some subcubes by others is not always necessary.
\\\newline

%
\bibliographystyle{plain}
\bibliography{ref}

\end{document}